\documentclass[11pt]{article}

\usepackage[letterpaper, margin=1in]{geometry}

\usepackage{microtype}
\usepackage{amsthm,amsfonts,amsmath,amsfonts,amsmath,amssymb}
\usepackage{xspace}
\usepackage{algorithm,algorithmic}
\usepackage[textsize=tiny,disable]{todonotes}
\usepackage[shortlabels]{enumitem}
\usepackage{epsfig}
\usepackage{subcaption}
\usepackage{wrapfig}
\usepackage{times}
\usepackage{tikz}
\usetikzlibrary{decorations.pathmorphing}
\usetikzlibrary{calc}
\usepackage{authblk}
\usepackage{url}
\usepackage{xcolor}
\usepackage{mathtools}
\usepackage{amsthm}
\usepackage{capt-of}
\usepackage{soul}
\usepackage{enumitem}
\usepackage{layouts}
\usepackage[]{mdframed}
\usepackage{tabularx}
\usepackage{ragged2e}
\usepackage{setspace}
\usepackage{thmtools,thm-restate}

\DeclareMathAlphabet{\mathbbold}{U}{bbold}{m}{n}

\newtheorem{theorem}{Theorem}

\newtheorem{definition}[theorem]{Definition}
\newtheorem{problem}[theorem]{Problem}
\newtheorem{lemma}[theorem]{Lemma}

\newtheorem{remark}[theorem]{Remark}

\ifdefined\final
\newcommand{\comment}[1]{}
\newcommand{\studentnote}[1]{}
\newcommand{\rrnote}[1]
\newcommand{\ignore}[1]{}
\else
\newcommand{\comment}[1]{\textsl{\small[#1]}\marginpar{\tiny\textsc{To Do!}}}
\newcommand{\ignore}[1]{}
\fi

\newcommand{\LC}{\text{LocalCharge}}
\newcommand{\GC}{\text{UniqueGlobalCharge}}
\newcommand{\TSGC}{\text{CommonGlobalCharge}}

\newcommand\Algitem{\medskip\par\noindent$\bullet$\ \justifying}

\usepackage{hyperref}

\hypersetup{
    colorlinks=true,
    linkcolor=blue,
    filecolor=magenta,      
    urlcolor=cyan,
    citecolor=red,
    pdftitle={Overleaf Example},
    pdfpagemode=FullScreen,
    }
\title{Putting Off the Catching Up:
Online Joint Replenishment Problem with Holding and Backlog Costs}
\author[]{Benjamin Moseley}
\author[]{Aidin Niaparast}
\author[]{R. Ravi\thanks{This material is based upon work supported in part by the Air Force Office of Scientific Research under award number FA9550-23-1-0031, Google Research Award, an Infor Research Award, a Carnegie Bosch Junior Faculty Chair, NSF grants CCF-2121744 and CCF-1845146 and ONR Grant N000142212702.}}
\affil[]{Tepper School of Business, Carnegie Mellon University, USA. \texttt{{moseleyb,aniapara,ravi}@andrew.cmu.edu}}
\date{}
\begin{document}
\maketitle
\begin{sloppypar}
\begin{abstract}
    We study an online generalization of the classic Joint Replenishment Problem (JRP) that models the trade-off between ordering costs, holding costs, and backlog costs in supply chain planning systems. 
    A retailer places orders to a supplier for multiple items over time: each request is for some item that the retailer needs in the future, and has an arrival time and a soft deadline. 
    If a request is served before its deadline, the retailer pays a holding cost per unit of the item until the deadline. 
    However, if a request is served after its deadline, the retailer pays a backlog cost per unit. 
    Each service incurs a fixed joint service cost and a fixed item-dependent cost for every item included in a service. These fixed costs are the same irrespective of the units of each item ordered. 
    The goal is to schedule services to satisfy all the online requests while minimizing the sum of the service costs, the holding costs, and the backlog costs.  
    
   Constant competitive online algorithms have been developed for two special cases: the make-to-order version when the deadlines are equal to arrival times~\cite{buchbinder2013online}, and the make-to-stock version with hard deadlines with zero holding costs~\cite{bienkowski2014better}. Our general model with holding and backlog costs has not been investigated earlier, and no online algorithms are known even in the make-to-stock version with hard deadlines and non-zero holding costs.
   We develop a new online algorithm for the general version of online JRP with both holding and backlog costs and establish that it is 30-competitive.  Along the way, we develop a 3-competitive algorithm for the single-item case that we build on to get our final result. Our algorithm uses a greedy strategy and its competitiveness is shown using a dual fitting analysis.
\end{abstract}
\section{Introduction}

The \emph{JRP} is a fundamental problem in supply chain management ~\cite{levi2,LRS06,khouja2008review,PengWangWang,buchbinder2013online,bienkowski2014better}, and it has been studied both in the offline and online models.

\paragraph{Offline JRP.} In the offline version of the JRP (see~\cite{levi2} for pointers to relevant work), there is a set of $n$ commodities (also called items) that a retailer stocks and sells over a planing horizon $T$. The demand for each item at each time period is assumed to be known (either as a given set of deterministic quantities or via a description of the demand distribution for each item over time). The demand for an item at a given time $t$ must be fulfilled by units of the item ordered at or before $t$, i.e., no backlogging (delaying the satisfaction of the demand) is allowed. 

With no costs for ordering, the retailer will simply order as many units of each item as is demanded as they arise over time. 
Ordering items, however, incurs costs for the retailer: every order irrespective of the set of items ordered incurs a joint ordering cost $c(r)$; additionally, if units of item $v$ are part of the order, there is an item-dependent cost of $c(v)$ irrespective of the number of units of item $v$ ordered.

With nonzero ordering costs, the retailer will then order the total demand for all items at the beginning of the planning horizon in a single order, but this ignores practical constraints and costs in storing these items.
 JRP models this with a holding cost function $h^v_{tt'}$ for holding one unit of item $v$ for the period $t$ (when it is ordered) to $t'$ (when it is used to satisfy demand). 
 In this paper, we assume that the holding cost is $h$ per unit time for all items, so that $h^v_{tt'} = h \cdot(t' - t)$  for all items $v$. The goal of JRP is to find an ordering policy that minimizes the sum of ordering costs and holding costs. This version has also been called the {\em make-to-stock} JRP because items are ordered to be stored until they are consumed. 

In the {\em make-to-order} version of JRP (studied in~\cite{buchbinder2013online}), holding inventory is not allowed. Hence, the demand for each item at $t$ can only be satisfied by items ordered at or after $t$. Instead of holding costs, there are backlog or delay costs ($b^v_{tt'}$) to satisfy the orders as above, which must be minimized in addition to the ordering costs. 
The make-to-stock and make-to-order versions of JRP are equivalent in their offline versions~\cite{LRS06}.

The bulk of the research on JRP assumes either offline deterministic demand, where the demands for each item and each time period are known in advance, or stochastic demand, where the probability distributions of future demands are known (see the surveys ~\cite{khouja2008review,PengWangWang}). However, these assumptions may be unrealistic in practical scenarios, and motivated the study of online algorithms~\cite{borodin2005online} for the problem, which assumes future demands are completely unknown and generated by an adversary.

\paragraph{Online JRP with holding and backlog.}
Buchbinder et al.~\cite{buchbinder2013online} study the online version of make-to-order JRP (with no holding), and give an elegant 3-competitive algorithm using the online primal-dual framework. They also provide a lower bound of approximately 2.64 on the competitiveness of any deterministic online algorithm for the problem.
On the other hand, for the online version of the make-to-stock version of the JRP (with no backlogs) in the special case where holding costs are zero, Bienkowski et al.~\cite{bienkowski2014better} designed a 2-competitive algorithm and showed that this competitive ratio was optimal.

In reality, the problem faced by retailers involves ordering for items jointly, while allowing for both holding inventory and backlogging some demand when it is not available. In particular, if we wish to model the adversarial (non-stochastic) demand model online, there is a specific time (arrival time) when the retailer is cognizant of the demand for a future day (as a result of either a planning exercise or a bulk order with a specific target delivery date). This naturally sets up the online arrival of demand information for each item, which is associated with an arrival time, a deadline, and associated costs for early or late fulfillment. 
In this paper, we model and study the online version of this general variant of JRP. 

Motivated by the scheduling literature~\cite{Pinedo_2012}, we assume that each demand has an arrival time and a deadline. 
The arrival time can be thought of as the earliest time at which the requirement for a deterministic demand at the deadline materializes. 
Each demand can only be satisfied using services at times that are not earlier than its arrival time. 
Our model captures the scenario where a retailer anticipates at time $a$ that it is going to run out of item $v$ at time $d \geq a$. So, it can only send a request after the arrival time $a$ to the supplier for the item $v$ to try and meet the deadline $d$.
We note that papers addressing the traditional offline versions of the JRP do not consider the notion of an arrival time for demands and assume instead that they are known at the beginning of the planning horizon (or equivalently, that their arrival time is zero). With the introduction of the arrival time in the offline models, the equivalence of the make-to-order and make-to-stock versions (described in~\cite{LRS06}) no longer holds. However, it is natural and important to introduce this notion of an arrival time in the online formulation.

In the definition below, we follow the literature on multi-level aggregation problems~\cite{chrobak2014online} and model the demand requests as arising in the leaves of a tree with two levels: the root node represents the joint ordering cost and each leaf represents an item or commodity with its node cost being the item-specific ordering cost.
\begin{problem}[Online Joint Replenishment Problem]
\label{prob:Online JRP}
    Let $T$ be a two-level tree with root $r$ and leaf nodes $v_1,\ldots,v_n$.
    Each node $v \in V(T)$ has a cost $c(v)$ associated with it: the root cost represents the joint ordering cost (also called fixed service cost), and the leaf costs represent the item-specific ordering costs (also called fixed item costs).
    
    A set of requests arrive at different leaves over time, with possibly more than one request arriving at the same time. Each \textit{request} $\rho$ is specified by a tuple ${\rho=(v_\rho, a_\rho, d_\rho)}$, where $v_\rho$ is the leaf node corresponding to the commodity that makes the request, $a_\rho$ is the \textit{arrival time} of the request, and $d_\rho \geq a_\rho$ is the (soft) \textit{deadline} of the request. 
    We assume each request is for one unit of the commodity since we can replicate demands for more units accordingly. Moreover, we assume that there is no ordering cost associated with each unit of each item\footnote{This is because such costs must be paid by any solution and they are the same among different solutions; hence they can be ignored.}.
    In the online version, at each time $t$, only the requests with arrival time before or at time $t$ are known, while nothing is known about future requests. In particular, the total number of requests and the length of the time horizon of the problem are unknown.
    
    Each request can only be satisfied at or after its arrival time, but possibly after its deadline. The holding and backlog costs are assumed to be linear and uniform among different requests. In particular, if a request $\rho$ is satisfied at a time $t \in [a_\rho,d_\rho]$, we incur a holding cost of $h \cdot (d_\rho - t)$, and if it is satisfied at a time $t \in (d_\rho,\infty)$, we incur a backlog cost of $b \cdot(t-d_\rho)$.

    The requests must be satisfied using a series of \textit{services} (the term we use for orders in this paper). A service is a subtree of $T'$ of $T$ that contains $r$, and can satisfy any subset of as yet unsatisfied requests made by nodes in $V(T')$. The service cost is $\sum_{v \in V(T')} c(v)$, which is the sum of the fixed joint cost of the service (root) and the costs of the items (leaves) included in the service. The problem is to minimize the total cost of servicing all the demand, which is the sum of the service costs, holding costs, and backlog costs.
\end{problem}

\subsection{Solution Approaches and Our Contribution}
The online make-to-order JRP with only backlog cost that is studied in~\cite{buchbinder2013online} is naturally suited to the online primal-dual approach surveyed in~\cite{buchbinder2009design}. In the algorithm suggested in~\cite{buchbinder2013online}, once a new order for some item arrives, the dual variables associated with satisfying that demand start to grow over time, and when a dual constraint becomes tight, a service is triggered. However, this approach cannot be used for cases where a demand can be satisfied before its deadline (i.e., when holding inventory is allowed), and completely fails if a demand can not be satisfied after a deadline (i.e., if backlog costs are infinite). 

On the other hand, the (make-to-stock) JRP with deadlines model (JRP-D) studied in~\cite{bienkowski2014better} is also considerably simpler to handle. When a deadline triggers service, all active requests can be satisfied due to the simplifying zero holding costs assumption. Their algorithm then uses future deadlines in increasing order to accumulate more items until the joint ordering cost has been spent to preemptively satisfy some items. Then, a simple analysis shows that any optimal schedule must satisfy requests at least at half the rate (in costs) as this algorithm.

To understand the fundamental difference between handling general holding cost and backlog cost, consider the special case of the problem with only one item (i.e. one leaf). When there is only backlog cost (i.e., no holding is allowed), once a service is triggered at time $t$, the optimal decision is to include all the requested units of the item that have arrived before or at $t$ in that service, as more backlog cost is incurred if the algorithm satisfies them using some later service. Therefore, the algorithm only needs to decide the times at which services are triggered. 
However, when there is only holding cost (i.e., no backlog is allowed), the algorithm can wait until the deadline of the first unsatisfied request to trigger a service since an earlier service would be wasteful and would incur additional holding costs. But when a service is triggered, the algorithm needs to decide which unsatisfied requests with deadlines in the future it wants to include in this service. The additional decision of which active requests to include in each service is a fundamental difficulty in handling holding costs. 
When both holding and backlog are allowed, the algorithm faces both challenges of when to trigger service and which active requests to fulfill preemptively at this service.

To appreciate this, consider the single-item problem with just backlog: the optimal deterministic algorithm waits until it accumulates backlog cost $s$, where $s$ is the service cost. This accumulated backlog easily ``justifies'' this service, i.e., the dual variables associated with these requests can each be set to the amount of their backlog, which sums up to $s$.
This way, it is easily proved that this algorithm is 2-competitive.
With just holding costs, the situation is already more complicated. This is because when we trigger a service, there might not be enough active requests with deadlines in the future to fulfill by paying their holding cost up to value $s$, thus disallowing the local accumulation of dual values for this service.
The problem becomes even more algorithmically challenging when there are multiple items (i.e. multiple leaves) as in the general JRP, as the algorithm has to also decide not only additional requests but also which subset of additional {\em items} it wants to include in the current service. 

Our main contribution is a constant competitive algorithm for this problem, the first known algorithm with non-trivial worst-case guarantees. Note that the results of ~\cite{buchbinder2013online,bienkowski2014better} imply a lower bound of 2.754 for the competitive ratio of deterministic online algorithms even for the special version with no holding costs. 
To the best of our knowledge, no online algorithm is currently known even for the very special case of Online JRP with only one item and only holding costs and a strict deadline (no backlog allowed).
Along the way to developing the full algorithm, we illustrate our key ideas for the general singe-item case.
\begin{theorem}
\label{thm:single-edge}
    There exists a 3-competitive deterministic polynomial-time algorithm for Online JRP with backlog and holding costs with a single item.
\end{theorem}
We then extend this to the general multi-item case with a worse competitive ratio.

\begin{theorem}
    \label{thm:online JRP}
    There exists a 30-competitive deterministic polynomial-time algorithm for Online JRP with backlog and holding costs.
\end{theorem}

The algorithm is loosely inspired by the primal-dual approach for constructing approximation algorithms~\cite{AKR95,GW95,buchbinder2009design}.
The main idea of the algorithm is to trigger services using accumulated greedy dual variable values of a natural linear programming formulation for the JRP. The service is triggered when these greedily growing dual variables (corresponding to constraints requiring serving the requests in the primal) become tight for some dual constraint. 
Despite this connection, the algorithm itself does not construct any dual solution explicitly but is stated as a simple greedy scheme. In contrast,~\cite{buchbinder2013online}
explicitly maintains and updates dual values in the algorithm.

As in a traditional primal-dual method, the duals corresponding to requests are increased simultaneously at the rate of backlog costs once their deadlines are past and time advances.
Broadly, the algorithm waits until enough backlog cost is accumulated in a first backlog phase to trigger a service (and makes tight a constraint on the dual values summing up to no more than the service costs). 
To ensure feasibility of dual-fitting, we aggregate and serve some overdue requests for additional items in a second backlog phase.
Finally, we also aggregate and serve requests with deadlines in the future for all included items in both phases. This ensures that the service cost, the backlog costs, and the holding costs paid are balanced. We use dual-fitting to define a feasible dual solution of value within a constant factor of the total algorithm's costs to prove the competitive ratio.

In more detail, the algorithm accumulates backlog costs for requests for a specific item (whose deadlines have passed) until it `fills' the item-specific ordering cost and the item becomes `mature'. Mature items then `overflow' their `surplus backlog costs' towards filling the cost of the root (the joint ordering cost) in a `Mature Backlog Phase'. A service is triggered when the root is fully paid for by these accumulating backlog costs.
At a service, all overdue requests from all mature items are satisfied.  In addition, in a `Premature Backlog Phase', a few more premature items that are close to becoming mature are accumulated in roughly increasing order of their time to maturity as a preemptive step (similar to the `simulation' step in the algorithm of~\cite{buchbinder2013online}). Then, in a `Local Holding Phase', for each included item in the service, requests with future deadlines that can be satisfied by paying an overhead up to the item-dependent ordering cost are preemptively included, again in increasing order of their deadline, in this service. Finally, a second preemptive `Global Holding Phase' includes additional requests in increasing order of future deadlines among all included items by paying an overhead of at most the joint ordering cost. The details of the four phases are in Section~\ref{sec:multi}. The bulk of our proof argues that the preemptive phases allow us to fit a feasible dual of value at least a constant fraction of the total costs paid by the triggering services.

We prove the single-item result in the next section before providing the algorithm for the general multi-item case in the following. 

\begin{remark}
    Although for simplicity throughout the paper we assume the holding and backlog cost functions are linear, all the results easily extend to the more general case where holding and backlog cost functions are non-negative and increasing (as long as they are the same across different items). 
    In particular, the results still hold when the holding and backlog costs between times $s$ and $t$, denoted by $h_{st}$ and $b_{st}$ respectively, satisfy the following property: 
    for each $s' \leq s \leq t \leq t'$ we have $h_{st} \leq h_{s't'}$ and $b_{st} \leq b_{s't'}$.
    These general results will appear in a journal version of this paper.
\end{remark}
\subsection{Related Work}
\label{sec:rel_work}
\nopagebreak
As noted above, offline JRP is studied in both the make-to-stock model, where only holding is allowed (but no backlogs), and the make-to-order model, where only backlog is allowed (but no holding). 
The special case of the make-to-order model where the delay function is linear is called JRP-L.
The special case of the make-to-stock model where there are hard deadlines but zero holding costs is called JRP-D.

\paragraph{Offline Approximation Ratios.} The offline JRP is known to be strongly $\mathbb{NP}$-hard~\cite{arkin1989computational, becchetti2009latency, nonner2009approximating}, even for the special cases of JRP-D and JRP-L. Also, JRP-D is $\mathbb{APX}$-hard~\cite{nonner2009approximating, bienkowski2015approximation}.

On the positive side, Levi, Roundy and Shmoys~\cite{LRS06} gave the first approximation algorithm for (offline make-to-stock) JRP, with a ratio of 2 under a general notion of holding costs. Their algorithm uses a primal-dual approach and can incorporate both holding and backlog costs.
This ratio was subsequently improved to 1.8
by Levi et al.~\cite{levi2008constant, levi2006improved}. Later, Bienkowski et al.~\cite{bienkowski2014better} improved the approximation ratio to 1.791. 
For offline (make-to-stock) JRP-D, the special case with zero holding costs, the ratio was reduced to 5/3 by Nonner
and Souza~\cite{nonner2009approximating} and then to $\approx 1.574$ by Bienkowski et al.~\cite{bienkowski2015approximation}.

The special case of the make-to-order JRP with a single item is also known as \textit{Dynamic TCP Acknowledgement}. 
This problem was first introduced by Dooly, Goldman, and Scott \cite{dooly1998tcp, dooly2001line} in 1998. 
In the offline setting, they gave a dynamic programming solution that solves the problem optimally in $O(n^2)$ time.

\paragraph{Online Competitive Ratios.} For the online version of (make-to-stock) JRP-D,  Bienkowski et al.~\cite{bienkowski2014better} designed a 2-competitive online algorithm and prove that it has an optimal competitive ratio.

Brito et al.~\cite{brito2012competitive}  gave a 5-competitive algorithm for online make-to-order JRP. 
Buchbinder et al.~\cite{buchbinder2013online} gave a 3-competitive algorithm for this problem. They also proved a lower bound of 2.64 for this problem, which was later improved to 2.754 in~\cite{bienkowski2014better}. These lower bounds apply to the more restricted version of the problem, where the delay functions are linear.

For the online version of the single-item make-to-order JRP (Dynamic TCP acknowledgment), Dooly et al.\cite{dooly1998tcp, dooly2001line} gave a 2-competitive algorithm, and presented a matching lower bound for any deterministic algorithm.  
Seiden~\cite{seiden2000guessing} proved an $\frac{e}{e-1}$ lower bound in the online setting for any algorithm. Later, Karlin, Kenyon and Randall~\cite{karlin2001dynamic} designed a randomized algorithm achieving $\frac{e}{e-1}$ competitive ratio.
As we mentioned before, the single-item variant of the problem with just holding was not studied before in the online setting. 

The make-to-order and make-to-stock versions of JRP are closely related to different Online Aggregation and Inventory Routing problems, respectively, which have been studied extensively. 
We review additional related literature on these problems in Appendix~\ref{sec:rel_app}.
\section{Warm Up: Single-Item Case}
\label{sec:single-item}
In this section, we study the special case of Online JRP where there is only one item and prove Theorem~\ref{thm:single-edge}.
The underlying tree only consists of a root $r$ and a leaf $v$. The cost of each service is always $s:=c(r)+c(v)$. 
First, we present an algorithm for single-item Online JRP. 
Then we write a natural linear program relaxation for the problem and obtain its dual. 
We then use dual fitting to prove Theorem~\ref{thm:single-edge}, i.e., we come up with a feasible dual solution whose objective value is at least $\frac{1}{3}$ the cost of our algorithm. This proves that the algorithm is 3-competitive.

\subsection{Algorithm}
In this section, we describe our online algorithm for single-item online JRP. 
We begin by describing the algorithmic challenges.

Suppose that only one request has arrived for this single item so far with deadline $d$.
Clearly, it is not optimal to trigger a service before $d$, as we incur an unnecessary holding cost. 
One option is to trigger a service at $d$, which results in no holding or backlog cost. 
This is indeed optimal if no other request arrives in the near future.
However, in case some other requests come shortly after $d$, it would be beneficial to aggregate all requests together and serve them all at once using one service after all their deadlines.
This approach prevents paying a service cost for each of them, at the expense of paying backlog costs for the overdue requests. 

Since the algorithm does not have access to future requests, it should hedge its bets between these two cases.   
When there is no holding allowed, the single-item problem is equivalent to the TCP-acknowledgment problem, for which Dooly et. al.~\cite{dooly1998tcp} show that the best deterministic approach is to {\bf wait until the sum of the backlog costs of the overdue requests accumulates to be equal to the service cost $s$} (similar to the familiar ski-rental idea in online algorithms) and then trigger a service. 
Our proposed algorithm follows the same stopping rule to trigger a service. 
Moreover, when a service is triggered, we satisfy all of overdue requests using this service since otherwise, when we fulfill them later, we will pay even more backlog cost for them.

Unlike the TCP-acknowledgment problem, after deciding when to trigger a service, we face an additional algorithmic challenge.
After serving overdue requests, there may still be \emph{active} requests with future deadlines. These requests can be fulfilled using the current service by paying the holding costs until their deadlines. The question we face is which subset of these active requests we must include in the current service.

On one extreme, assume that the algorithm does not fulfill any of these requests using the current service. 
Then, in the case where the backlog cost $b$ is large enough, the holding cost $h$ is small enough, and the deadlines of active requests are uniformly distributed in the future, the algorithm triggers a new service for each of the active requests.  Alternatively, the optimal algorithm satisfies all of them using the current service with small additional holding cost.
At the other extreme, the algorithm pays a significant holding cost to satisfy the active requests using the current service, which might be highly suboptimal as there is the option of waiting to get closer to the deadlines of those requests so they can be served with smaller holding/backlog costs.
To address this trade-off, our algorithm {\bf sets a budget of $s$ to pay for the holding costs of aggregating a subset of currently  active requests}.
The natural choice is to satisfy the requests with earlier deadlines first, as these requests will trigger a service sooner if we do not fulfill them now.

Putting these two phases together, we propose Algorithm~\ref{alg:single-item} for single-item online JRP.

\begin{algorithm}
    \caption{: Single-Item Online Joint Replenishment Problem}
    \label{alg:single-item}
    \setstretch{1.2}
    \vspace{0.1cm}
    \Algitem \textbf{Backlog Phase.} Wait until the sum of the backlog costs of unsatisfied overdue requests becomes $s$,
    the service cost. 
    Let $t$ be the time at which this happens. 
    Trigger a service $S$ at time $t$, and satisfy all of the unsatisfied requests with deadlines no later than $t$ using this service.

    \Algitem \textbf{Holding Phase.} Go over all the remaining active unsatisfied requests in order of their deadlines, and include them one by one in $S$, as long as the sum of their holding costs at time $t$ is at most $s$.\
    \vspace{0.1cm}
\end{algorithm}

\subsection{Analysis}
\label{sec:single-item-analysis}

In this section, we analyze the competitive ratio of Algorithm~\ref{alg:single-item}.
We start by writing a linear program for 
single-item Online JRP:

\begin{align}
    P=\min \qquad & \sum_{t} s\cdot z_t + \sum_{\rho} \left [ \sum_{a_\rho \leq  t \leq d_\rho} x_{\rho, t} \cdot h \cdot (d_\rho  - t) + \sum_{t > d_\rho} x_{\rho, t} \cdot b \cdot (t-d_\rho)
 \right ] &[\text{SJRP}_\text{P}]\\
    \text{s.t.} \qquad &\sum_{t \geq a_\rho}x_{\rho, t} = 1, & \forall \rho \label{const:request satisfaction}\\
	& x_{\rho, t} \leq z_t,  & \forall \rho , t \geq a_\rho \label{const:feasibility} \\
    & z_t, x_{\rho, t} \geq 0,  & \forall \rho , t \label{const:nonnegativity}
\end{align}

In [$\text{SJRP}_\text{P}$], $z_t$ represents whether a service is triggered at time $t$, and $x_{\rho, t}$ determines if request $\rho$ is satisfied at time $t$. Constraint~\eqref{const:request satisfaction} ensures that each request is satisfied eventually, and constraint~\eqref{const:feasibility} relates the two types of variables, making sure that a request is satisfied at time $t$ only if there is a service at that time.
Replacing constraint~\eqref{const:nonnegativity} with $z_t,x_{\rho, t} \in \{0,1\}$ results in an integer program whose optimum value is exactly $OPT$, the minimum cost incurred by the best offline algorithm. Since [$\text{SJRP}_\text{P}$] is a relaxation of this IP, it follows that $P \leq OPT$. Here is the dual of [$\text{SJRP}_\text{P}$], with variables $\alpha_\rho$ and $\beta_{\rho, t}$ assigned to constraints \eqref{const:request satisfaction} and \eqref{const:feasibility}, respectively:

\begin{align}
    D=\max \qquad & \sum_{\rho} \alpha_\rho & [\text{SJRP}_\text{D}]\\
    \text{s.t.} \qquad &\sum_{\rho:a_\rho \leq t}\beta_{\rho, t} \leq s, & \forall t \label{const:beta cap}\\
    & \alpha_\rho - \beta_{\rho, t} \leq h \cdot (d_\rho - t),  & \forall \rho , a_\rho \leq t \leq d_\rho \label{const:holding} \\
    & \alpha_\rho - \beta_{\rho, t} \leq b \cdot (t - d_\rho),  & \forall \rho , t > d_\rho \label{const:backlog} \\
    & \beta_{\rho, t} \geq 0,  & \forall \rho , t \label{const:dual nonnegativity}\\
    & \alpha_\rho \in \mathbb{R}, & \forall \rho
\end{align}
Here we give some intuition about how to approach dual fitting. 
The goal is to come up with a feasible solution to $[\text{SJRP}_\text{D}]$ with a large objective function $\sum_\rho \alpha_\rho$. 
Think of $\alpha$ variables as \emph{money} that we can use to \emph{pay} for the cost of our algorithm. 
When $\alpha_\rho>0$ for some request $\rho$, constraints~\eqref{const:holding} and \eqref{const:backlog} require us to assign nonnegative $\beta_{\rho,t}$ variables for $t \in [\max(a_\rho, d_\rho-\frac{\alpha_\rho}{h}), d_\rho + \frac{\alpha_\rho}{b}]$. 
So each $\alpha$ variable imposes a range of nonzero $\beta$ variables. 
Constraint~\eqref{const:backlog} can be rewritten as $\beta_{\rho, t} \geq \alpha_\rho - b \cdot (t - d_\rho)$, which means that at time $t=d_\rho$, we must set $\beta_{\rho,t}=\alpha_\rho$, and as we move past the deadline $d_\rho$ of $\rho$, the value of $\beta_{\rho,t}$ can be decreased at the rate of $b$ per unit of time, until it becomes 0 at time $t=d_\rho+\frac{\alpha_\rho}{b}$. 
We call this direction the \emph{forward direction} (of time).
In the other direction, which we call the \emph{backward direction}, as we move from $d_\rho$ towards $a_\rho$, Constraint~\eqref{const:holding} implies that we can decrease $\beta_{\rho,t}$ at the rate of $h$, until either it hits 0 at time $t=d_\rho-\frac{\alpha_\rho}{h}$, or hits $a_\rho$, whichever comes first. 
After that point, its value can stay at zero. Figure~\ref{fig:beta} illustrates these two cases.

In assigning these $\beta_{\rho,t}$ values as above over several requests $\rho$, the sum of the $\beta$ variables is bounded by Constraint~\eqref{const:beta cap}: at each time $t$, the total \emph{budget} we have for the sum of the $\beta$'s is at most the service cost $s$.

\begin{figure}[h!]
    \centering
    \begin{subfigure}{\linewidth}
      \centering
      \includegraphics[width=0.75\linewidth]{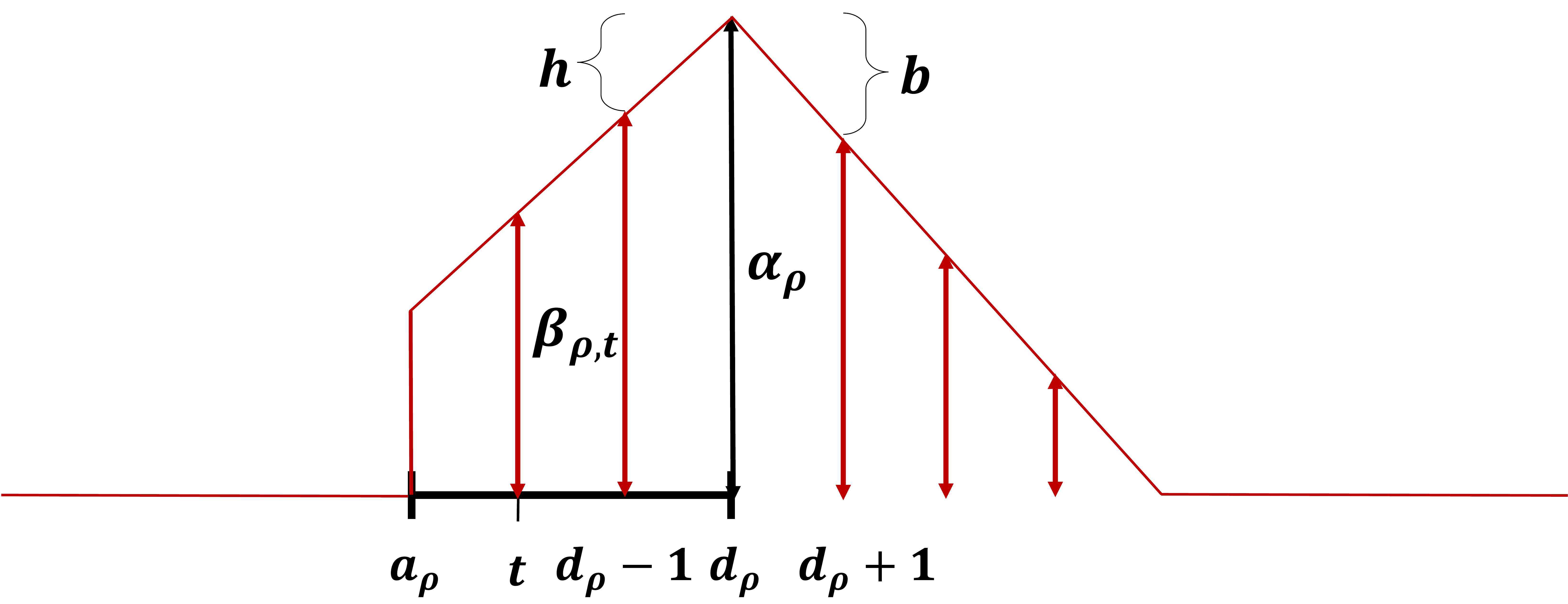}
      \caption{}
      \label{fig:beta(a)}
    \end{subfigure}
    \begin{subfigure}{\linewidth}
      \centering
      \includegraphics[width=0.75\linewidth]{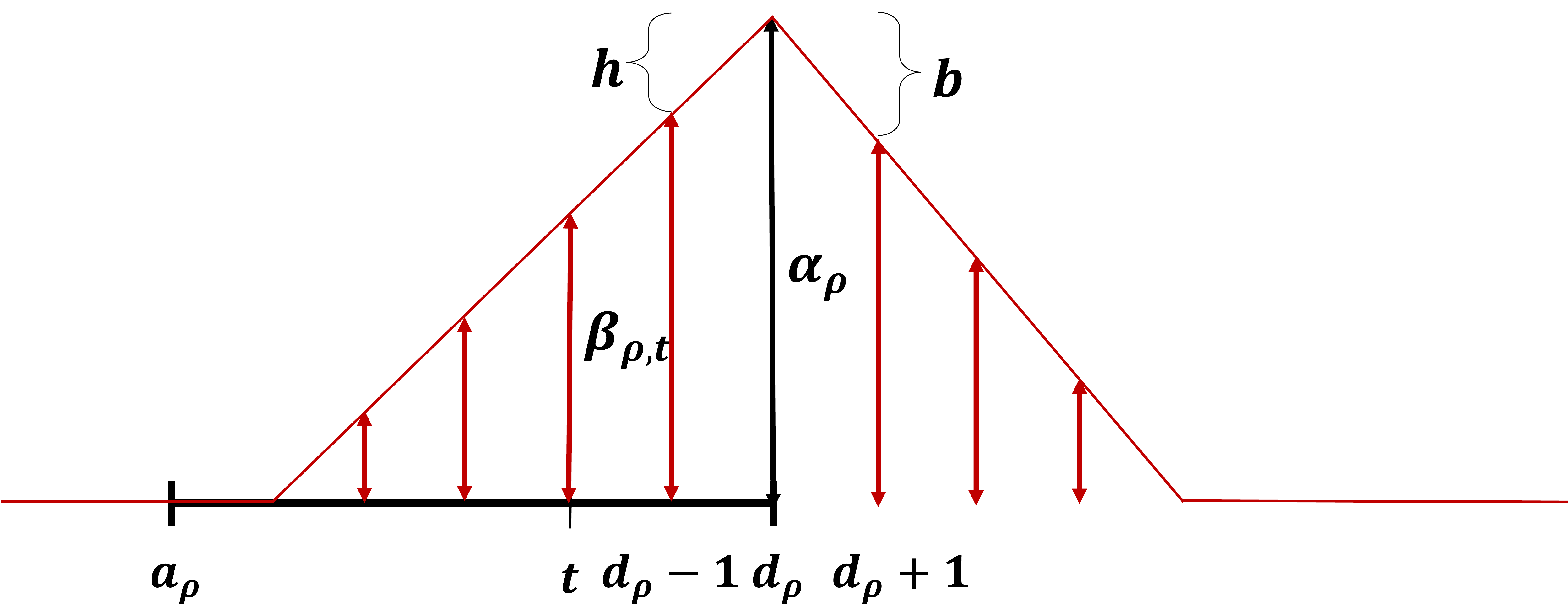}
      \caption{}
      \label{fig:beta(b)}
    \end{subfigure}
    \caption{Illustration of how the $\beta$ variables can depreciate in forward and backward directions starting from $d_\rho$ to satisfy constraints~\eqref{const:holding} and~\eqref{const:backlog} in $[\text{SJRP}_\text{D}]$ for some $\alpha_\rho>0$.
    Figures~\ref{fig:beta(a)} and \ref{fig:beta(b)} illustrate the cases where $d_\rho-\frac{\alpha_\rho}{h}< a_\rho$ and $d_\rho-\frac{\alpha_\rho}{h}> a_\rho$, respectively.}
    \label{fig:beta}
\end{figure}

Before describing the dual solution, we need some notation. 
Assume Algorithm~\ref{alg:single-item} triggers services $S_1,\ldots, S_N$ at times $t_1<\ldots<t_N$. Note that the algorithm does at most one service at each time step. 
For service $S_i$, let $B_i$ be the set of requests that trigger this service, i.e., the set of requests whose backlog costs sum up to $s$ at time $t_i$ and are satisfied in Backlog Phase of Algorithm~\ref{alg:single-item}. 
Also, let $H_i$ be the set of requests with deadlines after $t_i$ that get fulfilled by $S_i$ in Holding Phase of Algorithm~\ref{alg:single-item}.
By the design of the algorithm, we know $\sum_{\rho \in B_i} b \cdot (t_i-d_\rho) = s$ and $\sum_{\rho \in H_i} h \cdot (d_\rho - t_i) \leq s$.
For a service $S_i$, let $B^1_i$ be the set of requests in $B_i$ that were ``alive'' at the time of the previous service, i.e., $B^1_i:=\{\rho \in B_i: a_\rho \leq t_{i-1}\}$. Let $B^2_i:=B_i \backslash B^1_i$.
Since at each of the $N$ services, on top of the service cost, we pay an additional $s$ as backlog costs of the requests that trigger the service in Backlog Phase and at most additional $s$ to aggregate near-time future requests in Holding Phase, we have the following. 
\begin{lemma}
\label{lem:single-item alg cost}
    Algorithm~\ref{alg:single-item} provides a solution to the Single-Item Online JRP problem of cost at most $3Ns$.
\end{lemma}

To complete the analysis, we give a feasible solution for [$\text{SJRP}_\text{D}$] with an objective value of at least $Ns$. This would show that offline ${OPT\geq P \geq D \geq Ns}$, which then shows that the competitive ratio of  Algorithm~\ref{alg:single-item} is at most $\frac{3Ns}{Ns}=3$.

Even though our LPs have a variable or constraint for every point of time $t$, it is sufficient to only include the times corresponding to events in the optimal solution. 
Note that this is only for the sake of analysis since the algorithm itself does not use the LPs.

We begin by outlining the obstacles when developing a suitable dual solution and describe our strategies to address them.
The dual solution presented here is also a crucial building block of the analysis of Algorithm~\ref{alg:multiple-item} for the general Online JRP problem with multiple items described in Section~\ref{sec:algo}.

The goal is to assign a subset of requests to each service $S_i$ such that they can pay (accumulate) the cost $s$ using their variables $\alpha$, while all the associated variables $\beta$ across these subsets satisfy Constraint~\eqref{const:beta cap}.

Assume for a moment that for each service $S_i$, all of the requests that trigger this service arrive after the previous service, which has happened at time $t_{i-1}$, i.e., for each $\rho \in B_i$ we have $a_\rho > t_{i-1}$ (in this case, $B^1_i$ is empty for each $i$).
Then, for each service $S_i$ and for each $\rho \in B_i$,
we set $\alpha_\rho:=b \cdot (t_i-d_\rho)$ to be 
the backlog cost of $\rho$ at time $t_i$, i.e., the contribution of $\rho$ to triggering $S_i$.
With this assignment, the $\beta_{\rho,t}$ variables deflate to zero in the forward direction at time $d_\rho+\frac{\alpha_\rho}{b}=t_i$.
Moreover, in the backward direction, each $\beta_{\rho,t}$ becomes zero at $a_\rho>t_{i-1}$.
This is an ideal situation because of the following. 
\vspace*{-0.2cm}
\begin{enumerate}[itemsep=0.1cm,parsep=0.0cm]
    \item The sum of the $\alpha_\rho$ variables for all the requests that trigger $S_i$ is exactly $s$, as by definition, a service is triggered once the sum of the backlog costs of the active overdue requests becomes $s$;
    \item The $\beta$ variables associated with paying for a request $\rho \in B_i$ can have nonzero values only in the range $(t_{i-1}, t_i]$, which means that the $\beta$ variables for different services do not overlap.
    Since the $\alpha$ values for a service sum to $s$, this is also the largest sum of the corresponding $\beta$ values at any time in their interval, ensuring Constraint~\eqref{const:beta cap} is met. 
\end{enumerate}
\vspace*{-0.2cm}
Thus, in this scenario, a feasible dual solution with an objective value of $Ns$ can be easily constructed.

Now what if a request $\rho \in B_i$ arrives before $t_{i-1}$, i.e., $\rho \in B_i^1$?
We cannot do the same assignment for the dual variables, as the resulting nonzero $\beta$ variables for requests in different service sets $B_j$ can `pile up' on each other at the same time, resulting in a violation of Constraint~\eqref{const:beta cap}.
In this case, a key property of Algorithm~\ref{alg:single-item} helps us address this issue: since $\rho$ has arrived before $t_{i-1}$, this request was a candidate to be served in the Holding Phase of service $S_{i-1}$.
Since the algorithm in the Holding Phase probes the active requests in the order of their deadline, it means that service $S_{i-1}$ has served some other requests $H_{i-1}$ in its Holding Phase with earlier deadlines than $\rho$. 
Therefore, {\bf all the requests in $H_{i-1}$ have earlier deadlines than the requests in $B_i^1$}.
See Figure~\ref{fig:local charging} for an illustration.

\begin{figure}
    \centering
    \includegraphics[width=0.75\linewidth, height=2in]{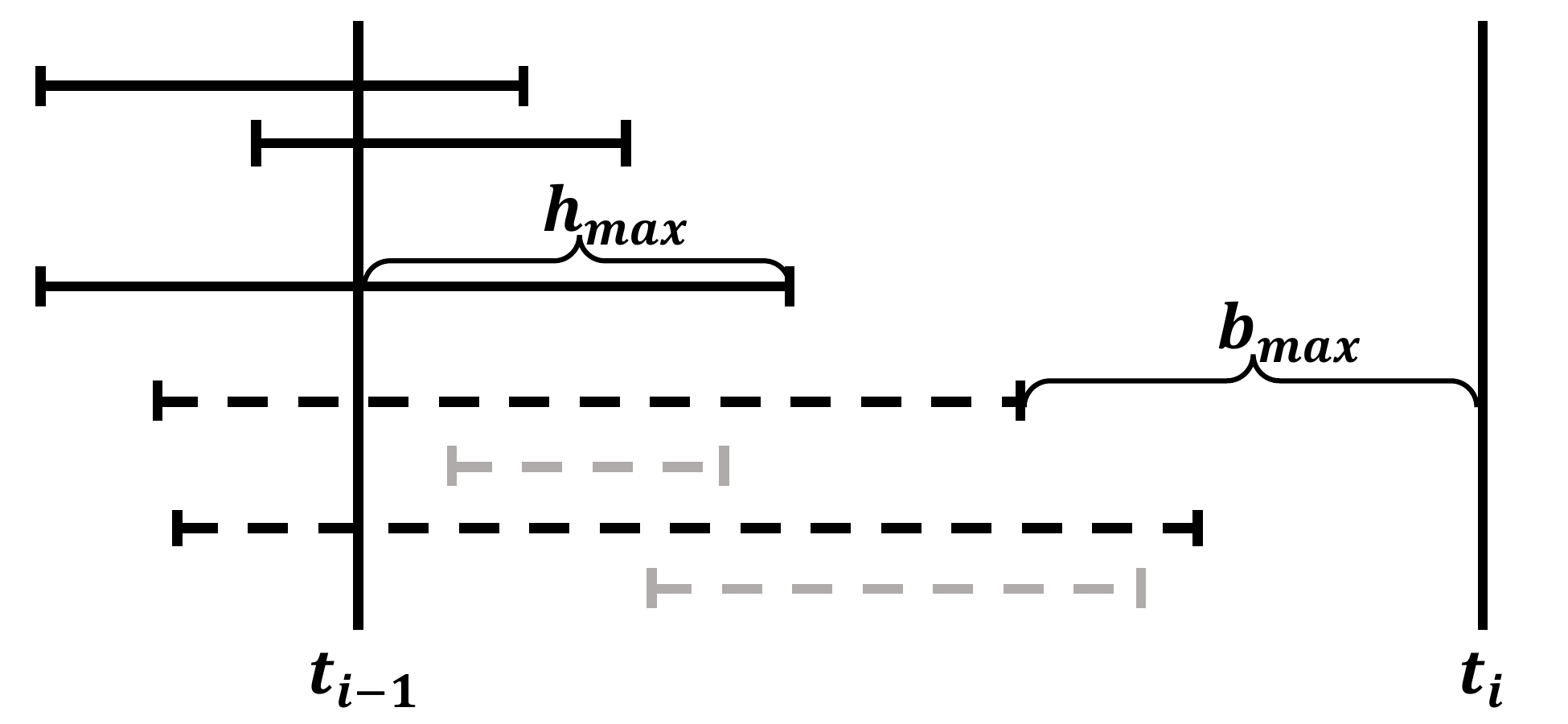}
    \caption{The requests in $H_{i-1}$ (solid), $B_i^1$ (dashed black), and $B_i^2$ (dashed grey) for some service $S_i$. These are the requests that are used to pay for (one-third of) the costs of service $S_i$.
    In this figure, $h_{\text{max}}$ is the holding cost of the request in $H_{i-1}$ with the latest deadline, and $b_{\text{max}}$ is the backlog cost of the request in $B_i^1$ which has the earliest deadline.}
    \label{fig:local charging}
\end{figure}

In this case, in addition to the requests in $B_i$, we also use the requests in $H_{i-1}$ to pay $s$.
Each request $\rho \in H_{i-1}$ can pay its holding cost that is incurred at time $t_{i-1}$ by setting $\alpha_\rho:=h \cdot (d_\rho-t_{i-1})$.
In this way, $\beta_{\rho,t}$ diminishes to zero in the backward direction at $d_\rho-\frac{\alpha_\rho}{h}=t_{i-1}$.
However, in the forward direction, it might go beyond $t_i$.

Remember that in order to control the sum of $\beta$ variables at each time (to satisfy Constraint~\eqref{const:beta cap}), it is crucial to keep the $\beta$ variables associated with different services completely disjoint.
This means that for each $\rho \in B_i \cup H_{i-1}$ and each $t \notin [t_{i-1},t_i]$, we want to have $\beta_{\rho,t}=0$.

To this end, in the dual solution, we compare $h_{\text{max}}$, the holding cost (from time $t_{i-1}$) of the request in $H_{i-1}$ with the latest deadline, with $b_{\text{max}}$, the backlog cost (at time $t_i$) of the request in $B_i^1$ with the earliest deadline (see Figure~\ref{fig:local charging}). 
Intuitively, if $h_{\text{max}}>b_{\text{max}}$, it means that the deadlines of the requests in $B_i^1$ are ``far enough'' from $t_{i-1}$ that their $\beta$ values will deflate to 0 in the backward direction before they reach $t_{i-1}$, making sure that all the nonzero $\beta$'s are confined within the range $[t_{i-1},t_i]$.
Similarly, if $h_{\text{max}}<b_{\text{max}}$, it follows that the deadlines of the requests in $H_{i-1}$ are ``far enough'' from $t_i$ that their $\beta$ variables will diminish to zero before $t_i$ in the forward direction.
Therefore, in the former case, we can use the requests in $B_i$ to pay $s$, and in the latter case, we can use the requests in $H_{i-1}$ to pay $s$. 

But what if the holding costs of the requests in $H_{i-1}$ are not enough to pay for $s$?
The final observation is that in such case, 
we can charge the requests in $B_i$ to pay the slack, as their $\beta$ values in the backward direction definitely deflate to 0 before reaching $t_{i-1}$.
This is because for each $\rho \in B_i^1$, since it was not satisfied at time $t_{i-1}$ despite being active then, there was not enough budget left in Holding Phase of $S_{i-1}$ to include it in $H_{i-1}$, which in turn means that its deadline has ``enough distance'' from $t_{i-1}$.

\paragraph{The Fitted Dual Solution.} 
Putting everything together, here is our proposed dual solution.
The proofs of Lemmas~\ref{lem:single-item dual objective} and~\ref{lem:single-item feasibility}, which show the dual solution has an objective value of $Ns$ and it is feasible, are deferred to Appendix~\ref{sec:appendix proofs}.
 
For each $i=1,\ldots,N$ do the following. 
Let $h_{\text{max}}$ be the holding cost (from time $t_{i-1}$) of the request in $H_{i-1}$ which has the latest deadline, and $b_{\text{max}}$ be the backlog cost (at time $t_i$) of the request in $B^1_i$ which has the earliest deadline (see Figure~\ref{fig:local charging}). 
Set $h_{\text{max}}:=0$ if $H_{i-1} = \emptyset$ and $b_{\text{max}}:=0$ if $B_i^1 = \emptyset$.
For ease of notation, set $H_0=\emptyset$.
There are two cases:
    \begin{itemize}
        \item \textbf{Case 1.} $h_{\text{max}} > b_{\text{max}}$: set $\alpha_\rho = b \cdot (t_i-d_\rho)$ for each $\rho \in B_i$, and set $\alpha_{\rho}=0$ for each $\rho \in H_{i-1}$. 
        \item \textbf{Case 2.} $h_{\text{max}} \leq b_{\text{max}}$: set $\alpha_\rho=h \cdot (d_\rho - t_{i-1})$ for each $\rho \in H_{i-1}$, and set $\alpha_\rho=\frac{s- h_{\text{sum}}}{s} \cdot b \cdot (t_i-d_\rho)$ for each $\rho \in B_i$, where $h_{\text{sum}}:=\sum_{\rho \in H_{i-1}} h\cdot (d_\rho-t_{i-1})$ is the sum of the holding costs of requests in $H_{i-1}$ at time $t_{i-1}$.
    \end{itemize}
For each $\rho$, set $\beta_{\rho,t}$ as follows (see Figure~\ref{fig:beta}):
\[
\beta_{\rho,t}= \begin{cases}
\max(0,\alpha_\rho - h \cdot (d_\rho - t)) &\text{for } a_\rho \leq t \leq d_\rho \\
\max(0,\alpha_\rho - b \cdot (t-d_\rho)) &\text{for } t > d_\rho.
\end{cases}
\]

\begin{lemma}
\label{lem:single-item dual objective}
    In the proposed dual solution, for each $i=1,\ldots,N$, we have $\sum_{\rho \in B_i \cup H_{i-1}} \alpha_\rho = s$. This shows that the dual objective function is $Ns$.
\end{lemma}

\begin{lemma}
\label{lem:single-item feasibility}
    The above dual solution is feasible for {\normalfont $[\text{SJRP}_\text{D}]$}.
\end{lemma}

\begin{proof}[Proof of Theorem~\ref{thm:single-edge}]
    Follows from Lemmas~\ref{lem:single-item alg cost},~\ref{lem:single-item dual objective}, and~\ref{lem:single-item feasibility}.
\end{proof}

\begin{remark}
    The analysis of the Algorithm~\ref{alg:single-item} is tight, i.e., there exists an instance of the problem for which the algorithm performs 3 times worse than the optimal offline algorithm:
    assume $b=h=1$ and all the arrival times are 0. 
    Assume there are $s$ deadlines at time 0, and for each $k \geq 1$, there are $2s$ deadlines at time $2k$. 
    Algorithm~\ref{alg:single-item} makes services at times $1,3,\ldots$ and pays $3s$ for each service. But one can do services at times $0,2,.\ldots$ and pay $s$ for each service (no holding or backlog cost incurred).
\end{remark}

\begin{remark}
    In the case where no backlog is allowed, i.e. $b=\infty$, the same analysis shows that Algorithm~\ref{alg:single-item} is 2-competitive. Note that in this case, the algorithm waits until one of the requests becomes due, and then triggers a service. After that, the algorithm goes over the active requests in the order of their deadlines and includes them in the current service until their total holding cost reaches $s$.
\end{remark}

\section{Online JRP with Multiple Items}
\label{sec:multi}
In this section, we describe an algorithm for (multiple-item) Online JRP. In Online JRP, each service includes a subset $U \subseteq \{v_1,\ldots,v_n\}$ of the items, and the cost of the service is the sum of the joint service cost $c(r)$, the item costs $c(U)=\sum_{u \in U}c(u)$, and the holding and backlog costs of the requests fulfilled using this service.

In the case where $c(r)>\sum_{i=1}^n c(v_i)$, the problem reduces to single-item JRP with service cost $c(r)$ as the sum of the joint service cost and item costs for each service is within a constant factor of $c(r)$.
On the other hand, if for some item $v$ we have $c(r) < c(v)$, we can decouple $v$ from the other items and treat it as a separate single-item problem, which results in the loss of only a constant factor.
While we make no assumptions, the above arguments show that the hard case is where $c(r) >c(v_i)$ for each $i$ and $\sum_{i=1}^n c(v_i)$ is much larger than $c(r)$.  Critically, in the multi-item case, the algorithm will need to justify paying both the joint service cost and the cost of each individual item that is included in a service.

\subsection{Algorithm}
\label{sec:algo}
We start with an intuitive description of the algorithm and then give a formal definition.

\smallskip
\noindent \textbf{Triggering a Service:} 
Using a greedy dual-growing approach, each item type $v_i$ is represented by a cup with a target volume $c(v_i)$ which is ``filled" over time with backlog costs of requests for this item.
Once an item's cup is full, the algorithm decides to include that item in the next service.
We call such items \emph{mature}.
The algorithm additionally accumulates  \emph{surplus} backlog cost, which can be thought of as the ``overflow'' of the item cups into another cup with volume $c(r)$ which corresponds to the joint service cost.
Once this joint cup is filled in a ``mature" backlog accumulation phase, the algorithm \emph{triggers} an actual service and includes all the mature items in the service, which we denote by $S$.
The algorithm naturally serves all the overdue requests for all items in $S$ since otherwise a larger backlog cost will be incurred to serve them in the future.

\smallskip
\noindent \textbf{Selecting Additional Items to Serve:}   It is not sufficient to include only the mature items as there are instances where the algorithm incurs large costs from requests for several items that are just about to become mature (see~\cite{buchbinder2013online}).  To overcome such instances, the algorithm also serves some of the \emph{premature} items; that is, the items that have not accumulated enough backlog cost to fill up their cups. Since we are adding premature items based on their backlog, this adds a second ``premature" backlog phase to the algorithm. In this phase, the algorithm sets a budget of $\Theta(c(r))$ and ``buys'' premature items one by one in the order of the first time they are going to be mature assuming no new request arrives for them.
It is worth mentioning that the algorithm presented in~\cite{buchbinder2013online} for make-to-order JRP has a similar step called the ``simulation step''.

\smallskip
\noindent \textbf{Selecting Additional Requests to Serve:}  The algorithm considers each item $v_i \in S$ and performs a ``local"  holding phase similar to Algorithm~\ref{alg:single-item} with a budget of $c(v_i)$. 
As explained later in Section~\ref{sec:technical_overview}, this helps us accumulate the $c(v_i)$ cost using dual variables.
Moreover, by the same logic, to gather the joint service cost $c(r)$ in a feasible dual, another ``global" holding phase is needed jointly among all the items in the service.
In this ``global'' holding phase, the algorithm sets a budget of $c(r)$ and goes over all the active requests for the items in $S$ in the order of their deadlines and fulfills them by paying their holding cost.

Formally, we propose Algorithm~\ref{alg:multiple-item} for Online JRP. 
Even though the algorithm uses a continuous notion of time for illustrative convenience, it is not hard to implement it in time polynomial in the number of requests since all relevant times of maturity are polynomially bounded by the number of requests.

\begin{algorithm}[h!]
    \caption{: Multiple-item Online Joint Replenishment Problem}
    \label{alg:multiple-item}
    \setstretch{1.2}
    \vspace{0.2cm}
    Start with time $t=0$ and continuously increase $t$. At each time $t$:
    
    \Algitem Let $\mathcal{A}_t$ be the set of all the \emph{active} requests at time $t$, i.e., the requests that have arrived before or at $t$ and are not satisfied by any of the previous services.
    
    \Algitem Assign a variable $b_t(v)$ to each item $v$, which indicates the amount of backlog accumulated by the unsatisfied overdue requests for $v$ at time $t$, i.e., $b_t(v)=\sum_{\rho \in \mathcal{A}_t: v_\rho=v, \ d_\rho < t} b \cdot (t-d_\rho)$. We call an item $v$ \emph{mature}, if $b_t(v) \geq c(v)$. Any backlog cost that is accumulated after $v$ becomes mature is called \emph{surplus backlog cost}.
    
    \Algitem Assign a variable $b_t(r)$ to the root, which indicates the sum of the surplus backlog costs for the mature items at time $t$, i.e., $b_t(r) = \sum_{v: b_t(v) \geq c(v)} (b_t(v) - c(v))$.
    
    \Algitem \textbf{Mature Backlog Phase} If $b_t(r)=c(r)$, i.e., the surplus backlog cost accumulated equals the joint service cost, trigger a service $S$ at time $t$. Include all the mature items $v$ in $S$, satisfy all their overdue requests using $S$, and remove them from $\mathcal{A}_t$.
    
    \Algitem \textbf{Premature Backlog Phase.} For each \emph{premature} item $v$ that has at least one unsatisfied request, compute the first time $\text{mature}_t(v)$ it is going to be mature, assuming no new request shows up, i.e., $\text{mature}_t(v)$ is a time $t'$ such that $\sum_{\rho \in \mathcal{A}_t: v_\rho=v, \ d_\rho < t'} b \cdot (t'-d_\rho) = c(v)$.  
    Sort all of the premature items in non-decreasing order of $\text{mature}_t(v)$, and include them in $S$ one by one, as long as the sum of their item costs $c(v)$ is at most $2c(r)$. 
    For each newly included item, satisfy all of their overdue requests using $S$, and remove them from $\mathcal{A}_t$.
    
    \Algitem \textbf{Local Holding Phase.} For each item $v$ included in $S$, iterate over the remaining active requests for $v$ in non-decreasing order of their
    deadlines, and satisfy them one by one using $S$, as long as the sum of their holding cost is at most $c(v)$ (Note that since all the overdue requests in $v$ are already satisfied using one of the backlog phases, all the remaining active requests in $v$ have deadlines in the future).

    \Algitem \textbf{Global Holding Phase.} Go over the remaining active requests for items included in $S$ in non-decreasing order
    of their deadlines, and satisfy them one by one using $S$, as long as the sum of their holding cost is at most $c(r)$.
    \vspace{0.2cm}
\end{algorithm}
\subsection{Primal and Dual Problems}

We use dual-fitting to analyze the algorithm.
First, we start by writing a linear program and its dual for Online JRP.

\begin{align}
    P=\min \ & \sum_{v\in \{r,v_1,\ldots,v_n\}}\sum_t c(v)\cdot z_{v,t} &[\text{JRP}_\text{P}] \\
    &+\sum_{\rho} \left [ \sum_{a_\rho \leq  t \leq d_\rho} x_{\rho, t} \cdot h \cdot (d_\rho  - t) + \sum_{t > d_\rho} x_{\rho, t} \cdot b \cdot (t-d_\rho)\right ] & \\
    \text{s.t.} \ &\sum_{t \geq a_\rho}x_{\rho, t} = 1, & \forall \rho \label{const JRP_P:request satisfaction}\\
	& x_{\rho, t} \leq z_{v_\rho,t},  & \forall \rho , t \geq a_\rho \label{const JRP_P:feasibility} \\
    & z_{v_i,t} \leq z_{r,t}, & \forall t, 1 \leq i \leq n \label{const JRP_P:parent}\\
    & x_{\rho, t} \geq 0,  & \forall \rho , t \label{const JRP_P:x nonnegativity} \\
    & z_{v,t} \geq 0, & \forall v,t \label{const JRP_P:z nonnegativity}
\end{align}

In $[\text{JRP}_\text{P}]$, $z_{v,t}$ represents whether a service is made at time $t$ that includes node $v$, and $x_{\rho, t}$ determines if  request $\rho$ is satisfied at time $t$. 
Constraint~\eqref{const JRP_P:request satisfaction} ensures that each request is eventually satisfied, and Constraint~\eqref{const JRP_P:feasibility} relates the two types of variables, ensuring that a request is satisfied at time $t$ only if a service is provided at that time. 
Constraint~\eqref{const JRP_P:parent} guarantees that the root is included in all services, i.e., the joint service cost is paid for each service that is triggered.
If the variables are constrained to be binary, the resulting Integer Program will be an exact formulation for the problem. Thus $[\text{JRP}_\text{P}]$ is a relaxation of the original problem, and it follows that $P \leq OPT$. 

The dual of $[\text{JRP}_\text{P}]$ has variables $\alpha_\rho$, $\beta_{\rho, t}$, and $\gamma_{v_i,t}$ assigned to the constraints \eqref{const JRP_P:request satisfaction},\eqref{const JRP_P:feasibility}, and \eqref{const JRP_P:parent}, respectively:

\begin{align}
    D=\max \quad & \sum_{\rho} \alpha_\rho & [\text{JRP}_\text{D}]\\
    \text{s.t.} \quad &\sum_{\rho:a_\rho \leq t, v_\rho = v_i}\beta_{\rho, t} - \gamma_{v_i,t}\leq c(v_i), & \forall t, 1 \leq i \leq n\label{const JRP_D:beta cap}\\
    &\sum_{i=1}^n \gamma_{v_i,t} \leq c(r), &\forall t \label{const JRP_D:gamma cap}\\
    & \alpha_\rho - \beta_{\rho, t} \leq h \cdot (d_\rho - t),  & \forall \rho , a_\rho \leq t \leq d_\rho \label{const JRP_D:holding} \\
    & \alpha_\rho - \beta_{\rho, t} \leq b \cdot (t - d_\rho),  & \forall \rho , t > d_\rho \label{const JRP_D:backlog} \\
    & \beta_{\rho, t} \geq 0,  & \forall \rho , t \label{const JRP_D:beta nonnegativity}\\
    & \gamma_{v_i,t} \geq 0, & \forall t,1 \leq i \leq n  \label{const JRP_D:gamma nonnegativity}
\end{align}

In $[\text{JRP}_\text{D}]$, variables $\beta_{\rho,t}$ are only defined for $t \geq a_\rho$. In the analysis, for simplicity, we define $\beta_{\rho,t}$ for all values of $t$ and set $\beta_{\rho,t}$ to be 0 for each $t<a_\rho$.

Here we give some intuition about the dual problem. 
The goal is to find a feasible solution to $[\text{JRP}_\text{D}]$ with a large objective function $\sum_\rho \alpha_\rho$. 
As in the single-item case, we think of $\alpha$ variables as \emph{money} that can be used to \emph{pay} for the cost of our algorithm. 
When $\alpha_\rho$ is increased for some request $\rho$, Constraints~\eqref{const JRP_D:holding} and \eqref{const JRP_D:backlog} force a feasible solution to have non-negative $\beta_{\rho,t}$ variables for $t \in [\max(a_\rho, d_\rho-\frac{\alpha_\rho}{h}), d_\rho + \frac{\alpha_\rho}{b}]$. 
Each nonzero $\alpha$ variable thus imposes a range of nonzero $\beta$ variables. The sum of the $\beta$ variables for requests for each item $v_i$ is bounded by Constraint~\eqref{const JRP_D:beta cap}. 
Assume that the $\gamma$ variables are all zero: then, Constraint~\eqref{const JRP_D:beta cap} insists that at each time $t$, the total amount of \emph{budget} that we have for the sum of the $\beta$'s for requests $\rho$ associated with item $v_i$ is at most the cost of the item $c(v_i)$.
This can be thought of as a \emph{local budget} that we have for each item per unit of time. 
If the cost incurred by $\alpha$ variables is beyond this local budget, then the variables $\gamma$ must be used, which represent a \emph{global budget} that is shared between all items. 
Constraint~\eqref{const JRP_D:gamma cap} implies that at each time $t$, the total amount of the global budget is at most the cost of the joint service $c(r)$. 

\subsection{Notation}
Let $S_1,\ldots,S_N$ be the services that our algorithm triggers, and assume that they happen at times $t_1\leq \ldots \leq t_N$. 
Let $V_{i,1}$ and $V_{i,2}$ be the set of items included in service $S_i$ in Mature Backlog Phase and Premature Backlog Phase, respectively, and define $V_i:=V_{i,1} \cup V_{i,2}$.
Let $V_{i,A}$ be the set of items that have at least one active\footnote{Recall that an active request is one that has arrived before the current time and is as yet unsatisfied.} request immediately after the Mature Backlog Phase.
Note that $V_{i,2} \subseteq V_{i,A}$.
For ease of notation, we define $t_0:=0$ and $V_{0}=V_{0,A}:=\emptyset$.

Let $B_{i,1}$, $H_{i,1}$, and $H_{i,2}$ be the set of requests satisfied using $S_i$ in Mature Backlog Phase, Local Holding Phase, and Global Holding Phase, respectively.
For each item $v$, the requests for that item satisfied in these three phases are denoted by $B_{i,1}^v$, $H_{i,1}^v$, and $H_{i,2}^v$, respectively.
For each item $v\in V_{i,2}$, let $B_{i,2}^v$ be the set of requests for $v$ that contribute to making $v$ mature at time $\text{mature}_{t_i}(v)$, i.e., the set of active requests $\rho$ at the beginning of Premature Backlog phase such that $a_\rho \leq t_i$, $d_\rho < \text{mature}_{t_i}(v)$, and $v_\rho=v$.
Let $B_{i,2}:=\bigcup_{v \in V_{i,2}} B_{i,2}^v$.
Note that in Premature Backlog Phase, only the requests in $B_{i,2}$ that have deadlines before $t_i$ are satisfied. 
So $S_i \subseteq B_{i,1} \cup B_{i,2} \cup H_{i,1} \cup H_{i,2}$ (we abuse the notation here and use $S_i$ to refer to both the $i$'th service and the set of requests included in the $i$'th service). 
For a subset $U$ of vertices in the tree, define $c(U):=\sum_{u \in U} c(u)$.

Consider a service $S_i$ and an item $v\in V_{i,1}$. The item $v$ is included in the Phase I backlog of service $S_i$, so it has matured at some time $t \leq t_i$. This implies that the sum of the backlog costs of the active requests for $v$ at time $t$ is at least $c(v)$. 
Since active requests for $v$ at time $t$ are all included in $B_{i,1}^v$, and $t \leq t_i$, we conclude that $\sum_{\rho \in B_{i,1}^v} b \cdot (t_i-d_\rho) \geq c(v)$. 
Let $\rho_1,\ldots,\rho_m$ be the requests in $B_{i,1}^v$ in increasing order of arrival time. Let $k^*$ be the smallest $k$ such that $\sum_{l=1}^k b \cdot (t_i-d_{\rho_l}) \geq c(v)$. Define $L_i^v:=\{\rho_1,\ldots,\rho_{k^*}\}$. If $\sum_{l=1}^{k^*} b \cdot (t_i-d_{\rho_l}) = c(v)$, define $R_i^v:=\{\rho_{k^*+1},\ldots,\rho_m\}$; otherwise define $R_i^v:=\{\rho_{k^*},\ldots,\rho_m\}$. Note that $L_i^v \cup R_i^v = B_{i,1}^v$ and 
$L_i^v \cap R_i^v$ is either $\emptyset$ or $\{\rho_{k^*}\}$. In the latter case, the item becomes mature in the middle of processing the dual value for request $\rho_{k^*}$ in this order.

Just to repeat, similar to the single-item case, when $\alpha_\rho>0$ for some request $\rho$, constraints~\eqref{const JRP_D:holding} and \eqref{const JRP_D:backlog} require us to assign nonnegative $\beta_{\rho,t}$ variables for $t \in [\max(a_\rho, d_\rho-\frac{\alpha_\rho}{h}), d_\rho + \frac{\alpha_\rho}{b}]$; see Figure~\ref{fig:beta}. 
As we move from time $d_\rho$ to later or earlier times, we can decrease the value of $\beta_{\rho,t}$.
We call the former \emph{forward direction}, and the latter \emph{backward direction} (of time).

\subsection{Challenges and a High-Level Description of the Dual Solution}
\label{sec:technical_overview}
In each service $S_i$, all costs are within a constant factor of $c(V_{i,1})+c(r)$; see Lemma~\ref{lem: alg cost}. In general, neither of these terms dominates the other.
Therefore, we will fit a dual solution that accumulates both $\Theta(c(V_1))$ and $\Theta(c(r))$.

In the single-item case, to pay $s$ for each service $S_i$, the requests in $H=H_{i-1}$ and $B=B_i$ are used.
There are  3 ingredients that are utilized in the dual solution for the single-item case:
\begin{enumerate}
[label*=\textbf{(I\arabic*)}]
    \item\label{ing1} Sum of the backlog costs of the requests in $B$ at time $t_i$ is $s$.
    \item\label{ing2} 
    The active requests at $t_{i-1}$ with future deadlines are added to $H$ in order of their deadline and the budget for their holding costs is $s$. 
    \item\label{ing3} Each request in $B$ that has arrived before $t_{i-1}$ is a candidate in the holding phase of $S_{i-1}$.
\end{enumerate}

In the dual solution presented in Section~\ref{sec:dual solution} for $[\text{JRP}_{\text{D}}]$, we use a  similar dual solution to the single-item case to partially pay the costs of Algorithm~\ref{alg:multiple-item}.
In fact, all the cost $c(V_{i,1})$ for each service $S_i$ is paid in this way, since all 3 ingredients needed are present.
Assume $v\in V_{i,1}$ and let $S_\ell$ be the last service before $S_i$ that includes $v$.
Let $B:=B_{i,1}^v$ and $H:=H_{\ell,1}^v$.
Item $v$ is included in $V_{i,1}$ because the sum of backlog accumulated for it is at least $c(v)$, i.e., the sum of the backlog costs of the request in $B$ at time $t_i$ is $c(v)$ (Ingredient~\ref{ing1}).
Moreover, in the Local Holding Phase of $S_\ell$, Algorithm~\ref{alg:multiple-item} sets a budget of $c(v)$ to go over the active requests for $v$ with deadlines in the future and serves them using $S_\ell$ in the order of their deadlines (Ingredient~\ref{ing2}).
Lastly, each request in $B$ that has arrived before $t_\ell$ is a candidate in the Local Holding Phase of $S_\ell$ for item $v$ (Ingredient~\ref{ing3}).
We can accumulate a dual value of $c(v)$ using a  similar approach to the dual assignment for the single-item case.
This is done in step~\ref{D1} of the dual assignment in Section~\ref{sec:dual solution} using Lemma~\ref{lem:local charging}.
We call this procedure {\bf local charging}, as it is used to recover the item costs (which can be thought of as the local costs of the algorithm) only using the local budget and the $\beta$ variables. In particular, the $\gamma$ variables in $[\text{JRP}_{\text{D}}]$ are not used in the local chargings.

Paying $c(r)$ is more nuanced.
At first glance, it might seem that for two consecutive services $S_{i-1}$ and $S_i$, we can follow the same structure and set $B$ to be the set of requests that caused the surplus backlog of $c(r)$ in $S_i$, and set $H:=H_{i-1,2}$, the set of requests served in the global holding phase of $S_{i-1}$.
In fact, with this definition of $B$ and $H$, we have ingredients~\ref{ing1} and~\ref{ing2}; but~\ref{ing3} might be missing. 
This is because the set of \emph{items} included in the services $S_{i-1}$ and $S_i$ are not necessarily the same, i.e., there might be a request $\rho \in B$ that has arrived before $t_{i-1}$ but its item $v_\rho$ is not included in the service $S_{i-1}$\footnote{This difficulty does not arise in the single item case since all requests are for the same item and they are considered in the previous service if they were active.}. 
This means the request was not considered in the Global Holding Phase of $S_{i-1}$. 
However, this shows us that we can pay for the surplus backlog costs of the requests that were among the candidates in Global Holding Phase of $S_{i-1}$, i.e., the requests for the items in $V_{i,1} \cap V_{i-1}$ (This is done in step~\ref{D222} of the dual assignment using Lemma~\ref{lem:two-sided global charge}).

To recover the remainder of the surplus backlog cost, namely for the items in $V_{i,1} \backslash V_{i-1}$,  we need a new idea. 
This is where the Premature Backlog Phase of Algorithm~\ref{alg:multiple-item} is exploited.

In the Premature Backlog Phase of service $S_{i-1}$, for each premature item $v$, the algorithm computes the first time $\text{mature}_{t_{i-1}}(v)$ that $v$ becomes mature using only active requests at time $t_{i-1}$.
Suppose $\text{mature}_{t_{i-1}}(v_1)\leq \ldots \leq \text{mature}_{t_{i-1}}(v_m)$.
The algorithm iterates over the nodes $v_1,\ldots,v_m$, in this order, and includes them one by one in service $S_{i-1}$ by paying their item costs $c(v)$, as long as the sum of the item costs is at most $2c(r)$. 
Assume that the algorithm includes items $V_{i-1,2}=\{v_1,\ldots,v_k\}$ in the Premature Backlog Phase of service $S_{i-1}$, and for simplicity, assume $\sum_{j=1}^k c(v_j)=\Theta(c(r))$.

There are two cases to consider:
\begin{itemize}
    \item $t_i \geq \text{mature}_{t_{i-1}}(v_{k+1})$. In this case, for each item $v \in \{v_1,\ldots,v_k\}$, using Lemma~\ref{lem:local charging}, we pay for $c(v_i)$.  This allows us to pay $\sum_{j=1}^k c(v_j) = \Theta(c(r))$.
    The only difference between this case and the local charging is that we set $B$ to be equal to the set of requests that make $v$ mature at time $\text{mature}_{t_{i-1}}(v)$.
    We set $H=H_{\ell,1}^v$ where $\ell$ is the last service before $i-1$ that includes $v$.
    This case corresponds to step~\ref{D21} in the dual assignment below.
    The fact that $\text{mature}_{t_{i-1}}(v) \leq \text{mature}_{t_{i-1}}(v_{k+1}) \leq t_i$ ensures that the non-zero $\beta$ variables set in Lemma~\ref{lem:local charging} are restricted within $[t_\ell,t_i]$.  This prevents the $\beta$ variables used to pay for different services from accumulating at some fixed time.
    
    However, there is one important detail that we need to be careful about:
    some of the requests that contribute to $v$ becoming mature at time $\text{mature}_{t_{i-1}}(v)$ have deadlines after $t_{i-1}$, which means that they are not served using $S_{i-1}$ (and they might still be active at time $t_i$). 
    This means that we are charging a request at time $t_{i-1}$ that has not yet been served.
    To address this issue, we show that each request is charged at most {\em twice} during the local chargings in the construction of the dual solution.
    
    \item $t_i < \text{mature}_{t_{i-1}}(v_{k+1})$. Note that we only need to pay for the surplus backlog cost of the items in $V_{i,1}\backslash V_{i-1}$, which does not include any of the items in $V_{i-1,2}=\{v_1,\ldots,v_k\}$ since $V_{i-1,2}\subseteq V_{i-1}$.
    Let $v \in V_{i,1} \backslash V_{i-1}$.
    We show that all the requests that contribute to the surplus backlog cost of $v$ have arrived after time $t_{i-1}$ (Lemma~\ref{lem:global charge arrival}). Consider setting $\alpha_\rho$ for each such request to be equal to the backlog cost of $\rho$ at time $t_i$. In this case, we know that the corresponding $\beta$ variables can only have nonzero values in $[t_{i-1},t_i]$.
    Lemma~\ref{lem:global charge arrival} holds because either $v$ does not have any active requests at time $t_{i-1}$, or $v$ is in $\{v_{k+1},\ldots,v_m\}$. This means that the active requests at time $t_{i-1}$ can only accumulate enough backlog cost to make $v$ mature at $\text{mature}_{t_{i-1}}(v_{k+1}) > t_i$. Thus, all the requests responsible for the surplus backlog cost of $v$ at time $t_i$ have arrived after $t_{i-1}$.
    This case corresponds to~\ref{D221} in the dual assignment.
\end{itemize}

Finally, we mention one other detail that needs to be taken into account. 
In the above, for an item $v \in V_{i,1}$, we talk about requests that contribute to the surplus backlog cost of $v$ at time $t_i$.
Here we implicitly assume that the set of requests in $B_{i,1}^v$ is naturally partitioned into two sets: one set consisting of the earlier requests (in terms of arrival times) that make $v$ mature, and another set comprised of the later requests that are responsible for the surplus backlog cost.
However, in reality, a request can contribute to both the backlog cost that makes $v$ mature and the surplus backlog cost.
To partition the requests into two sets having the aforementioned properties, we ``artificially" partition the set $B_{i,1}^v$ into two sets $L_i^v$ and $R_i^v$ which play the roles of the two sets described above.
Note that this partitioning of the requests in $B_{i,1}^v$ is crucial as the charging used in~\ref{D221} can only be done on the requests in $R_{i}^v$ which we know have arrived after $t_{i-1}$.

Next we formally describe the dual solution.

\subsection{Dual Solution}
\label{sec:dual solution}
Our proposed dual solution is constructed in a modular way.
First, we describe the building blocks of the dual solution.
The (partial) dual solutions presented in the proofs of Lemmas~\ref{lem:local charging} and~\ref{lem:two-sided global charge} are very similar to our proposed dual solution for the single-item case.
The proofs are deferred to Appendix~\ref{sec:appendix proofs}.

\begin{restatable}[Local Charging Lemma for $v$ and $S_i$]{lem}{lemlocalcharging}
\label{lem:local charging}
    Let $S_i$ be a service that contains item $v$, and let $S_\ell$ be the last service before $S_i$ that includes $v$.
    Define $B$ and $t^*$ as follows:
    \begin{itemize}
        \item Case 1 (mature items): $v \in V_{i,1}$.
        Define $B:=L_i^v$, and $t^*:=t_i$.
        \item Case 2 (premature items): $v \in V_{i,2}$. Define $B:=B_{i,2}^v$ and let $t^*$ be the time $\text{mature}_{t_i}(v)$ defined in Premature Backlog Phase of service $S_i$ in Algorithm~\ref{alg:multiple-item}.
    \end{itemize}
    Also, let $H:=H_{\ell,1}^v$ be the set of requests for item $v$ that are satisfied in the Local Holding Phase of service $S_\ell$. 
    If $S_i$ is the first service that includes $v$, set $t_\ell:=0$ and $H:=\emptyset$. 
    There is an assignment for the variables $\alpha_\rho,\beta_{\rho,t}$ for each $\rho \in B \cup H$ such that:
    \begin{enumerate}[label=\upshape(\Roman*),ref=\thetheorem (\Roman*)]
        \item\label{cond:1} $\sum_{\rho \in B \cup H} \alpha_\rho = c(v),$
        \item\label{cond:2} $\beta_{\rho,t}=0$ for each $\rho \in B \cup H$ and $t \notin (t_\ell,t^*)$.  
        \item\label{cond:3} $\beta_{\rho,t} \leq \alpha_\rho$ for each $\rho \in B \cup H$ and each time $t$,
        \item\label{cond:4} $\alpha_\rho$ and $\beta_{\rho,t}$ for $\rho \in B \cup H$ satisfy constraints~\eqref{const JRP_D:holding},~\eqref{const JRP_D:backlog}, and~\eqref{const JRP_D:beta nonnegativity}.
        \item\label{cond:5} In Case 1, $\alpha_{\rho^*} \leq c(v) - \sum_{\rho \in L_i^v \backslash \{\rho^*\}}b \cdot (t_i-d_{\rho})$, where $\rho^*$ is the request in $L_i^v$ with the largest arrival time.
    \end{enumerate}
\end{restatable}    

\begin{restatable}[Two-Sided Global Charge for $S_i$]{lem}{lemtwosidedglobalcharging}
    \label{lem:two-sided global charge}
    Let $S_{i}$ be one of the services triggered by Algorithm~\ref{alg:multiple-item}, $i=2,\ldots,N$, and let $S_{i-1}$ be its previous service. 
    Let $H:=H_{i-1,2}$ and $B :=\bigcup_{v \in V_{i,1} \cap V_{i-1}} R_{i}^v$. 
    Then there is an assignment for the variables $\alpha_\rho$, $\beta_{\rho,t}$ and $\gamma_{v,t}$ for each $\rho \in H \cup B$ and $v \in V_{i,1} \cap V_{i-1} $ such that: 

    \begin{enumerate}[label=\upshape(\Roman*),ref=\thetheorem (\Roman*)]
        \item\label{cond two-sided:1} 
        $\sum_{\rho \in B \cup H} \alpha_\rho \geq 
        \sum_{v \in V_{i,1} \cap V_{i-1}}\left( \sum_{\rho \in B_{i,1}^v} b \cdot (t_i - d_\rho) - c(v)\right )$.
        \item\label{cond two-sided:2} $\beta_{\rho,t}=0$ and $\gamma_{v,t}=0$ for each $\rho \in B \cup H$, $v \in V_{i,1} \cap V_{i-1}$, and $t \notin (t_{i-1},t_i)$.  
        \item\label{cond two-sided:3} $\beta_{\rho,t} \leq \alpha_\rho$ for each $\rho \in B \cup H$ and each time $t$,
        \item\label{cond two-sided:4} $\alpha_\rho$, $\beta_{\rho,t}$ and $\gamma_{v,t}$ for each $\rho \in B \cup H$ and $v \in V_{i,1} \cap V_{i-1}$ satisfy all the constraints in $[\text{JRP}_\text{D}]$.
        \item\label{cond two-sided:5} 
        $\sum_{\rho:a_\rho \leq t, v_\rho = v}\beta_{\rho, t} = \gamma_{v,t}$,
        for each $v \in V_{i,1}\cap V_{i-1}$ and each time $t$.
    \end{enumerate}
\end{restatable}

We build the dual solution using the following three routines.

\begin{definition}[$\text{LocalCharge}(v,S_i)$]
    \label{def:local-charge}
    Assume $S_i$ is a service that includes item $v$.
    Do the dual assignment for item $v$ and service $S_i$ described in Lemma~\ref{lem:local charging}, and denote the variables by $\alpha^*$ and $\beta^*$. 
    For each request $\rho$ and time $t$, increase the current values of $\alpha_\rho$ and $\beta_{\rho,t}$ by $\frac{1}{4}\alpha^*_\rho$ and $\frac{1}{4}\beta^*_{\rho,t}$, respectively. 
\end{definition}

\begin{definition}[$\GC(\rho,S_i)$]
    \label{def:global-charge}
    Assume $\rho$ is a request in $B_{i,1}$. 
    Let $\alpha_\rho^0$ be the current value of $\alpha_\rho$, and define $\Delta:=b\cdot(t_i-d_\rho)-\alpha_\rho^0$. 
    Increase $\alpha_\rho$ by $\Delta$, and for each time $a_\rho \leq t \leq t_i$, increase $\beta_{\rho,t}$ and $\gamma_{v_\rho,t}$ by $\Delta$. 
\end{definition}

\begin{definition}[$\TSGC(S_i)$]
    \label{def:two-sided-global-charge}
    Let $\alpha^*$, $\beta^*$ and $\gamma^*$ be the dual assignment of Lemma~\ref{lem:two-sided global charge} for service $S_i$, where $i=2,\ldots,N$. 
    For each item $v$, request $\rho$, and time $t$, increase the values of $\alpha_\rho$, $\beta_{\rho,t}$, and $\gamma_{v,t}$ by $\frac{1}{2}\alpha^*_\rho$, $\frac{1}{2}\beta^*_{\rho,t}$, and $\frac{1}{2}\gamma^*_{v,t}$, respectively.
\end{definition}
    
Now we are ready to describe the dual solution. 

\paragraph{Dual Assignment.}
Consider a service $S_i$ for $i=2,\ldots,N$.
Without loss of generality, assume after doing Mature Backlog Phase for $S_{i-1}$, the set of premature items that have at least one active request is $V_{i-1,A}=\{v_1,\ldots,v_m\}$, and $\text{mature}_{t_{i-1}}(v_1)\leq \ldots \leq \text{mature}_{t_{i-1}}(v_m)$. 
Suppose in Premature Backlog Phase for $S_{i-1}$, the items $v_1,\ldots,v_k$ are included, i.e., $V_{i-1,2}=\{v_1,\ldots,v_k\}$. 
So either these are all the items with active requests, i.e., $m=k$, or we did not have enough budget to include $v_{k+1}$ in the service, i.e., $\sum_{j=1}^{k+1} c(v_j) > 2c(r)$. 
In the former case, define $\text{mature}_{t_{i-1}}(v_{k+1}):=\infty$.

Here is how we construct the dual solution for $[\text{JRP}_{\text{D}}]$.
Initially, set all the dual variables to zero. 
For each $i=1,\ldots,N$, do the following (note that $V_{0,A}=\emptyset$):
\begin{mdframed}[frametitle = {Dual Assignment $i$}, frametitlerule=true, roundcorner=5pt]
\label{Dual Assignment}
\begin{enumerate}
    [label*=\textbf{D.\arabic*}]
    \item\label{D1} 
    For each $v \in V_{i,1}$ do $\LC(v,S_i)$.
    \item\label{D2}There are two cases:
    \begin{enumerate}
    [label*=\textbf{.\arabic*}]
        \item\label{D21} \textbf{Case I}: $V_{i-1,A}\neq\emptyset$ and $t_i > \text{mature}_{t_{i-1}}(v_{k+1})$.

        For each node $v \in V_{i-1,2}$, do $\text{LocalCharge}(v,S_{i-1})$.
        
        \item\label{D22} \textbf{Case II}: $V_{i-1,A} = \emptyset$ or $t_i \leq \text{mature}_{t_{i-1}}(v_{k+1})$.

        Define $R_{i}^1:=\bigcup_{v \in V_{i,1} \backslash V_{i-1}} R_{i}^v$ and $R_{i}^2:=\bigcup_{v \in V_{i,1} \cap V_{i-1}} R_{i}^v$.
        \begin{enumerate}[label*=\textbf{.\arabic*}]
            \item\label{D221} For each $\rho \in R_i^1$, do $\GC(\rho,S_{i})$.
            \item\label{D222}
            If $R_i^2 \neq \emptyset$, do $\TSGC(S_i)$.
            \item\label{D223} If the total increase of the $\alpha$ variables in steps~\ref{D221} and~\ref{D222} is more than $c(r)$, scale down all the increases in these steps (for all the variables involved) so that the total increase of the $\alpha$ variables becomes exactly $c(r)$. 
        \end{enumerate}
    \end{enumerate}
\end{enumerate}
\end{mdframed}

Note that for $i=1$, since $V_{i-1,A}=V_{i-1}=\emptyset$, only steps~\ref{D1},~\ref{D221} and~\ref{D223} are called, where $R_i^1=\bigcup_{v \in V_{1,1}}R_1^v$.

Intuitively, Dual Assignment $i$ is responsible for paying a constant fraction of the costs of service $S_i$.

Before calculating the objective value of the above dual solution and proving its feasibility, we prove the following key structural lemma, which is used in the next two sections. 

\begin{lemma}
    \label{lem:global charge arrival}
    All the requests that are involved in step~\ref{D221} of Dual Assignment $i$, i.e., the requests in $R_{i}^1=\bigcup_{v \in V_{i,1} \backslash V_{i-1}} R_{i}^v$, arrive after $t_{i-1}$.
\end{lemma}

\begin{proof}
    In step~\ref{D2} of the Dual Assignment $i$, Case II happens when $V_{i-1,A}=\emptyset$ or $t_i<\text{mature}_{t_{i-1}}(v_{k+1})$. 
    Recall that $V_{i-1,A}=\{v_1,\ldots,v_m\}$ is the set of all the items that have at least one active request at time $t_{i-1}$ and are not included in Mature Backlog Phase in service $S_{i-1}$. 
    The times at which these items are going to mature, considering only their active requests at time $t_{i-1}$ immediately after Mature Backlog Phase, are $\text{mature}_{t_{i-1}}(v_1)\leq \ldots \leq \text{mature}_{t_{i-1}}(v_m)$. 
    Among these items, $V_{i-1,2}=\{v_1,\ldots,v_k\}$ are included in Premature Backlog Phase for service $S_{i-1}$.
    The requests involved in~\ref{D221} are the requests in $R_i^1=\bigcup_{v \in V_{i,1} \backslash V_{i-1}} R_{i}^v$.
    Since $\{v_1,\ldots,v_k\} = V_{i-1,2} \subseteq V_{i-1}$, none of the items $\{v_1,\ldots,v_k\}$ are in $R_i^1$.

    Let $v \in V_{i,1}\backslash V_{i-1}$ be one of the items involved in $R_i^1$. 
    We want to show that all the requests in $R_i^v$ have arrived after $t_{i-1}$.
    There are two cases:
    \begin{itemize}
        \item $v \in \{v_{k+1},\ldots,v_m\}$. 
        Since $t_i < \text{mature}_{t_{i-1}}(v_{k+1})$, we conclude that the sum of the backlog costs of the requests for $v$ that are active at time $t_{i-1}$  is less than $c(v)$ at time $t_i$, i.e., these requests are not enough to make $v$ mature at time $t_i$. 
        In Mature Backlog Phase of service $S_i$, only mature items are included. 
        So since $v \in V_{i,1}$, i.e., it was chosen in Mature Backlog Phase of service $S_i$, it is impossible for all the requests in $L_i^v$ to have arrived before $t_{i-1}$ (recall that the requests in $L_i^v$, by definition, can make $v$ mature at time $t_i$).
        In particular, it implies that the request with the latest arrival time in $L_i^v$ has arrived after $t_{i-1}$, which means that all the requests in $R_i^v$ have arrived after $t_{i-1}$.
        \item $v \notin V_{i,A}$. 
        This means that $v$ does not have any active requests at time $t_{i-1}$ after service $S_{i-1}$.
        Therefore all the requests for item $v$ that are served during Mature Backlog Phase of service $S_i$, i.e., $B_{i,1}^v$, have arrived after $t_{i-1}$.
        In particular, all the requests in $R_i^v \subseteq B_{i,1}^v$ have arrived after $t_{i-1}$.
    \end{itemize}
    Note that since all the requests arrive after time $0$, for $i=1$ we have that all the requests in $R_i^1$ arrive after $t_{i-1}=0$.
\end{proof}

\subsection{Dual Objective Value}
\label{sec: dual objective}

In this section, we calculate the objective value of the dual solution described in Section~\ref{sec:dual solution}, and compare it to the cost of Algorithm~\ref{alg:multiple-item}.

\begin{lemma}
    \label{lem:D21-increase}
    The total increase of the $\alpha$ variables if Step~\ref{D21} of Dual Assignment $i$ is invoked is at least $c(r)/2$.
\end{lemma}

\begin{proof}
When this case happens, it means that $v_{k+1}$ has some active requests at time $t_{i-1}$, i.e., $\text{mature}_{t_{i-1}}(v_{k+1})\neq \infty$. 
        But we have not included $v_{k+1}$ in $V_{i-1,2}$, which means that we could not afford to pay for $c(v_{k+1})$ in the Premature Backlog Phase of service $S_{i-1}$. 
        So $\sum_{\ell=1}^{k+1} c(v_\ell) > 2c(r)$. 
        On the other hand, $\text{mature}_{t_{i-1}}(v_{k+1})$ is the time at which $v_{k+1}$ becomes mature using only the backlog costs of the requests that are active at time $t_{i-1}$. 
        Since $v_{k+1}$ is not included in $S_{i-1}$, none of these requests are fulfilled using service $S_{i-1}$, and they remain active after this service.
        Thus, the actual time $v_{k+1}$ becomes mature after service $S_{i-1}$ is not more than $\text{mature}_{t_{i-1}}(v_{k+1}) < t_i$. 
        Since $S_i$ has happened at time $t_i$, it means that $v_{k+1}$ is included in $V_{i,1}$, as it was mature at time $t_i$. 
        So in~\ref{D1}, we have called $\LC(v_{k+1},S_i)$, which increases the sum of the $\alpha$ variables by $c(v_{k+1})/4$.
        We also call $\LC(v,S_i)$ for all $v \in V_{i-1,2}=\{v_1,\ldots,v_k\}$ in~\ref{D21}, which increases the sum of the $\alpha$ variables by $\sum_{\ell=1}^k c(v_\ell)/4$. 
        Therefore, $\alpha(S_i) \geq \sum_{\ell=1}^{k+1} c(v_\ell)/4 > 2c(r)/4 = c(r)/2$.
\end{proof}

\begin{lemma}
    \label{lem:D22-increase}
    The total increase of the $\alpha$ variables if step~\ref{D22} of Dual Assignment $i$ is invoked is at least $c(r)/2$.
\end{lemma}

\begin{proof}
    For $i=2,\ldots,N$, in step~\ref{D221} of Dual Assignment $i$, for each $\rho \in R_i^1=\bigcup_{v \in V_{i,1} \backslash V_{i-1}} R_{i}^v$, we call $\GC(\rho,S_i)$. 
        This function increases the value of $\alpha_\rho$ by $b \cdot (t_i-d_\rho)-\alpha_\rho^0$, where $\alpha_\rho^0$ is the value of $\alpha_\rho$ before calling $\GC$.
        Note that for each node $v \in V_{i,1} \backslash V_{i-1}$, only the requests in $L_i^v$ are used in the local chargings, which means that for each $\rho \in R_i^v \backslash L_i^v$, we have $\alpha_\rho^0 = 0$. 
        Also, $L_i^v \cap R_i^v$ can only have one member, and if $\rho^* \in L_i^v \cap R_i^v$, Lemma~\ref{lem:global charge arrival} shows that $a_{\rho^*}>t_{i-1}$, which means that $\rho^*$ was not used in any of the previous dual assignments, and the first time it was potentially used is in step~\ref{D1} of the current Dual Assignment $i$.
        In this case, $\LC(v,S_i)$ is called, which sets the value of $\alpha_{\rho^*}$ to $\frac{1}{4}\alpha_{\rho^*}^*$, where $\alpha_{\rho^*}^*$ is the $\alpha$ value assigned to $\rho^*$ in Lemma~\ref{lem:local charging}.
        Lemma~\ref{cond:5} ensures that $\alpha_{\rho^*}^* \leq c(v) - \sum_{\rho \in L_i^v \backslash \{\rho^*\}}b \cdot (t_i-d_{\rho})$, which means that
        \[
        \alpha_{\rho^*}^0 
        \leq \frac{1}{4}\left(c(v) - \sum_{\rho \in L_i^v \backslash \{\rho^*\}}b \cdot (t_i-d_{\rho})\right)
        \leq c(v) - \sum_{\rho \in L_i^v \backslash \{\rho^*\}}b \cdot (t_i-d_{\rho}).
        \] 
        Also, when $L_i^v \cap R_i^v=\emptyset$, we have $\sum_{\rho \in L_i^v}b\cdot(t_i-d_\rho)=c(v)$.
        Thus the total increase of the alpha variables in~\ref{D221} is 
        \begin{align*}
            \sum_{\rho \in R_i^1} (b \cdot (t_i-d_\rho)-\alpha_\rho^0) &= \sum_{v\in V_{i,1} \backslash V_{i-1}} \sum_{\rho \in R_i^v} (b \cdot (t_i-d_\rho)-\alpha_\rho^0) \\
            &=\sum_{v\in V_{i,1} \backslash V_{i-1}} \left [ \sum_{\rho \in R_i^v\backslash L_i^v} (b \cdot (t_i-d_\rho)-\alpha_\rho^0) + \sum_{\rho \in R_i^v \cap L_i^v} (b \cdot (t_i-d_\rho)-\alpha_\rho^0) \right ]\\
            &\geq \sum_{v\in V_{i,1}\backslash V_{i-1}} \left [ \sum_{\rho \in R_i^v\backslash L_i^v} b \cdot (t_i-d_\rho) + \sum_{\rho \in L_i^v} b \cdot (t_i-d_\rho) - c(v)\right ] \\
            &= \sum_{v\in V_{i,1}\backslash V_{i-1}} \left [ \sum_{\rho \in B_{i,1}^v} b \cdot (t_i-d_\rho) - c(v)\right ].
        \end{align*}
        
        In~\ref{D222}, we call $\TSGC(S_i)$, where due to Lemma~\ref{cond two-sided:1},
        increases $\sum_{\rho} \alpha_\rho$ by at least $\frac{1}{2}\sum_{v \in V_{i,1} \cap V_{i-1}}\left(\sum_{\rho \in B_{i,1}^v} b \cdot (t_i - d_\rho) - c(v)\right)$.
        
        Therefore, the total increase of the $\alpha$ variables in steps~\ref{D221} and~\ref{D222} of Dual Assignment $i$ is at least
        \begin{align*}
            &\sum_{v\in V_{i,1}\backslash V_{i-1}} \left [ \sum_{\rho \in B_{i,1}^v} b \cdot (t_i-d_\rho) - c(v) \right ]
            +
            \frac{1}{2}\sum_{v \in V_{i,1} \cap V_{i-1}} \left [ \sum_{\rho \in B_{i,1}^v} b \cdot (t_i - d_\rho) - c(v) \right ] \\
            &\geq 
            \frac{1}{2}
            \sum_{v \in V_{i,1}} \left [ \sum_{\rho \in B_{i,1}^v} b \cdot (t_i - d_\rho) - c(v) \right ] \\
            &= \frac{1}{2}c(r),
        \end{align*}
        where the last equality is derived by the fact that Algorithm~\ref{alg:multiple-item} has triggered $S_i$ in Mature Backlog Phase after it has accumulated exactly $c(r)$ surplus backlog cost.
\end{proof}

\begin{lemma}
    \label{lem:dual value}
    Let $\alpha(S_i)$ be the total increase of the $\alpha$ variables in Dual Assignment $i$. We have $\alpha(S_i) \geq \max(\frac{1}{4}c(V_{i,1}),\frac{1}{2}c(r))$ for each service $S_i$.
\end{lemma}

\begin{proof}
    Each time $\LC(v,S_i)$ is called for an item $v$ and a service $S_i$ that includes $v$, $\sum_\rho \alpha_\rho$ increases by $\frac{1}{4}\sum_\rho \alpha^*_\rho$, where $\alpha^*_\rho$ is the $\alpha$ variable assigned to $\rho$ in Lemma~\ref{lem:local charging}.
    By Lemma~\ref{cond:1}, we know that $\sum_\rho \alpha^*_\rho=c(v)$, which means that the increase in $\sum_\rho \alpha_\rho$ during $\LC(v,S_i)$ is exactly $\frac{1}{4} c(v)$.
    For $i=1,\ldots,N$, in Dual Assignment $i$, step~\ref{D1} is called, in which $\LC(v,S_i)$ is invoked for each $v \in V_{i,1}$.
    Thus, $\alpha(S_i) \geq \frac{1}{4}c(V_{i,1})$.

    Now we prove that $\alpha(S_i) \geq c(r)/2$. 
    Consider the two cases in~\ref{D2}:
    \begin{itemize}
        \item Case I: 
        In this case, by Lemma~\ref{lem:D21-increase}, in step~\ref{D21} the $\alpha$ variables increase by $c(r)/2$.
        \item Case II: In this case, by Lemma~\ref{lem:D22-increase}, in step~\ref{D22} the $\alpha$ variables increase by $c(r)/2$.
    \end{itemize}
\end{proof}

\begin{lemma}
    \label{lem: alg cost}
    The cost of the algorithm for service $S_i$ is at most $3c(V_{i,1})+9c(r)$.
\end{lemma}

\begin{proof}
    We break down the costs into three parts:
    \begin{itemize}
        \item Service Cost: We pay $c(r)$ as the joint service cost, and $c(V_{i,1})+c(V_{i,2})$ as the item-dependent costs. 
        From the design of the algorithm, we know $c(V_{i,2}) \leq 2c(r)$. So the service cost is at most $c(V_{i,1})+3c(r)$.
        \item Backlog Cost: When $S_i$ is triggered, it means that in Mature Backlog Phase we have $b_t(r)=c(r)$, which means that the surplus backlog cost accumulated after the items in $V_{i,1}$ have become mature is $c(r)$. 
        The backlog cost needed for each item $v\in V_{i,1}$ to become mature is $c(v)$. 
        So the total amount of backlog cost in Mature Backlog Phase is $c(V_{i,1}) + c(r)$.
        In Premature Backlog Phase, all of the new items that are picked up are premature, which means that the sum of the backlog costs for their requests at time $t_i$ is at most their item cost. 
        This means the sum of the backlog cost incurred in Premature Backlog Phase is at most $c(V_{i,2}) \leq 2c(r)$. 
        Thus the total backlog cost paid in service $S_i$ is at most $c(V_{i,1})+3c(r)$.
        \item Holding Cost: For each item $v$ in $V_{i,1} \cup V_{i,2}$, a holding cost of at most $c(v)$ is incurred in the Local Holding Phase of the algorithm.
        Moreover, a total cost of at most $c(r)$ is paid in the global holding phase.
        Therefore, the holding cost paid is at most $(c(V_{i,1})+c(V_{i,2}))+c(r)
        \leq c(V_{i,1})+3c(r)$.
    \end{itemize}
\end{proof}

\begin{lemma}
    \label{lem:cost ratio}
    The total cost incurred by the algorithm is at most $30 \sum_\rho \alpha_\rho$, for the proposed dual solution. 
\end{lemma}

\begin{proof}
    The proof follows immediately from Lemmas~\ref{lem:dual value} and \ref{lem: alg cost}.
\end{proof}

\subsection{Dual Feasibility}
In this section we show that the dual solution presented in Section~\ref{sec:dual solution} is feasible.
The proofs are deferred to Appendix~\ref{sec:appendix proofs}.

\begin{lemma}
    \label{lem:charged-twice}
    Each request $\rho$ is involved in at most two local chargings and one global charging, i.e., its corresponding variables $\alpha_\rho$, $\beta_{\rho,t}$ and $\gamma_{v_\rho,t}$ are modified at most two times during the calls to $\LC(.)$ and at most one time during the calls to $\GC(.)$ or $\TSGC(.)$.
\end{lemma}

\begin{lemma}
    \label{lem:local-beta}
    For each item $v$ and each time $t$, there are only two local chargings that can increase $\sum_{\rho:v_\rho=v}\beta_{\rho,t}$.
\end{lemma}

\begin{lemma}
    \label{lem:dual feasibility}
    The dual solution presented in Section~\ref{sec:dual solution} is feasible for [$\text{JRP}_\text{D}$].
\end{lemma}

\subsection{Proof of Theorem~\ref{thm:online JRP}}

\begin{proof}[Proof of Theorem~\ref{thm:online JRP}]
    For a particular instance $I$ of Online JRP, in Section~\ref{sec:dual solution}, we construct a solution to $[\text{JRP}_\text{D}]$, which by Lemma~\ref{lem:dual feasibility}, is feasible, and by Lemma~\ref{lem:dual value}, has an objective value of at least $\frac{1}{30} \text{Alg}(I)$, where $\text{Alg}(I)$ is the cost of Algorithm~\ref{alg:multiple-item} for instance $I$. 
    This objective value is a lower bound for $D$, the optimal value of $[\text{JRP}_\text{D}]$, which in turn, by weak duality, is a lower bound for $P$, the optimal value of $[\text{JRP}_\text{P}]$.
    Since $[\text{JRP}_\text{P}]$ is a relaxation of the original problem, we have the following:
    \[\frac{1}{30}\text{Alg}(I) \leq D \leq P \leq OPT(I),\]
    where $OPT(I)$ is the cost of the optimal solution for instance $I$.
    This shows that Algorithm~\ref{alg:multiple-item} is a 30-competitive algorithm for Online JRP.
\end{proof}
\section{Conclusion}

This paper considers the Joint Replenishment Problem (JRP) in the online setting with holding and backlog costs for the first time.  The main result is a new constant competitive greedy algorithm.  
Technically, this paper introduces a dual fitting approach that gives insight into the combinatorial structure of the problem.

There are several interesting directions for future work. One is to consider the case where each request can have a general monotonically increasing cost function for the holding and backlog costs. Appendix~\ref{app:nonuniform} shows that Algorithms~\ref{alg:single-item} and~\ref{alg:multiple-item} do not give bounded competitive ratios for this case, even when the cost functions are linear.
Another direction is to consider multilevel trees. There is a line of work on online algorithms for multilevel trees with hard deadlines or with just backlog costs (and no holding) \cite{azar2019general,bienkowski2020online,chrobak2014online}. 
The dual fitting approach here could give insights into improving this line of work and, further, can be useful for giving the first results for handling backlog and holding costs in arbitrarily deep trees. 
Finally, giving tight competitive ratios for different versions of Online JRP, either by strengthening the upper bounds or the lower bounds, is an intriguing open question.
There are no known lower bounds on the competitive ratio of the Online JRP model presented in this paper except for the ones that already hold for the case where just backlogging is allowed.

\bibliographystyle{abbrv}
\bibliography{biblio-2022}

\begin{thebibliography}{10}

\bibitem{AKR95}
A.~Agrawal, P.~Klein, and R.~Ravi.
\newblock {When trees collide : An approximation algorithm for the generalized Steiner problem on networks}.
\newblock {\em SIAM Journal on Computing}, 24(3):445--456, 1995.

\bibitem{arkin1989computational}
E.~Arkin, D.~Joneja, and R.~Roundy.
\newblock Computational complexity of uncapacitated multi-echelon production planning problems.
\newblock {\em Operations research letters}, 8(2):61--66, 1989.

\bibitem{azar2021online}
Y.~Azar, A.~Ganesh, R.~Ge, and D.~Panigrahi.
\newblock Online service with delay.
\newblock {\em ACM Transactions on Algorithms (TALG)}, 17(3):1--31, 2021.

\bibitem{azar2019general}
Y.~Azar and N.~Touitou.
\newblock General framework for metric optimization problems with delay or with deadlines.
\newblock In {\em 2019 IEEE 60th Annual Symposium on Foundations of Computer Science (FOCS)}, pages 60--71. IEEE, 2019.

\bibitem{becchetti2009latency}
L.~Becchetti, A.~Marchetti-Spaccamela, A.~Vitaletti, P.~Korteweg, M.~Skutella, and L.~Stougie.
\newblock Latency-constrained aggregation in sensor networks.
\newblock {\em ACM Transactions on Algorithms (TALG)}, 6(1):1--20, 2009.

\bibitem{bienkowski2020online}
M.~Bienkowski, M.~B{\"o}hm, J.~Byrka, M.~Chrobak, C.~D{\"u}rr, L.~Folwarczn{\`y}, {\L}.~Je{\.z}, J.~Sgall, N.~K. Thang, and P.~Vesel{\`y}.
\newblock Online algorithms for multilevel aggregation.
\newblock {\em Operations Research}, 68(1):214--232, 2020.

\bibitem{bienkowski2021new}
M.~Bienkowski, M.~B{\"o}hm, J.~Byrka, M.~Chrobak, C.~D{\"u}rr, L.~Folwarczn{\`y}, {\L}.~Je{\.z}, J.~Sgall, N.~K. Thang, and P.~Vesel{\`y}.
\newblock New results on multi-level aggregation.
\newblock {\em Theoretical Computer Science}, 861:133--143, 2021.

\bibitem{bienkowski2015approximation}
M.~Bienkowski, J.~Byrka, M.~Chrobak, N.~Dobbs, T.~Nowicki, M.~Sviridenko, G.~{\'S}wirszcz, and N.~E. Young.
\newblock Approximation algorithms for the joint replenishment problem with deadlines.
\newblock {\em Journal of Scheduling}, 18(6):545--560, 2015.

\bibitem{bienkowski2014better}
M.~Bienkowski, J.~Byrka, M.~Chrobak, {\L}.~Je{\.z}, D.~Nogneng, and J.~Sgall.
\newblock Better approximation bounds for the joint replenishment problem.
\newblock In {\em Proceedings of the twenty-fifth annual ACM-SIAM symposium on discrete algorithms}, pages 42--54. SIAM, 2014.

\bibitem{bienkowski2013online}
M.~Bienkowski, J.~Byrka, M.~Chrobak, {\L}.~Je{\.z}, J.~Sgall, and G.~Stachowiak.
\newblock Online control message aggregation in chain networks.
\newblock In {\em Workshop on Algorithms and Data Structures}, pages 133--145. Springer, 2013.

\bibitem{borodin2005online}
A.~Borodin and R.~El-Yaniv.
\newblock {\em Online computation and competitive analysis}.
\newblock cambridge university press, 2005.

\bibitem{bosman2020improved}
T.~Bosman and N.~Olver.
\newblock Improved approximation algorithms for inventory problems.
\newblock In {\em International conference on integer programming and combinatorial optimization}, pages 91--103. Springer, 2020.

\bibitem{bramel1995location}
J.~Bramel and D.~Simchi-Levi.
\newblock A location based heuristic for general routing problems.
\newblock {\em Operations research}, 43(4):649--660, 1995.

\bibitem{brito2012competitive}
C.~F. Brito, E.~Koutsoupias, and S.~Vaya.
\newblock Competitive analysis of organization networks or multicast acknowledgment: How much to wait?
\newblock {\em Algorithmica}, 64:584--605, 2012.

\bibitem{buchbinder2017depth}
N.~Buchbinder, M.~Feldman, J.~Naor, and O.~Talmon.
\newblock O (depth)-competitive algorithm for online multi-level aggregation.
\newblock In {\em Proceedings of the Twenty-Eighth Annual ACM-SIAM Symposium on Discrete Algorithms}, pages 1235--1244. SIAM, 2017.

\bibitem{buchbinder2013online}
N.~Buchbinder, T.~Kimbrel, R.~Levi, K.~Makarychev, and M.~Sviridenko.
\newblock Online make-to-order joint replenishment model: Primal-dual competitive algorithms.
\newblock {\em Operations Research}, 61(4):1014--1029, 2013.

\bibitem{buchbinder2009design}
N.~Buchbinder, J.~S. Naor, et~al.
\newblock The design of competitive online algorithms via a primal--dual approach.
\newblock {\em Foundations and Trends{\textregistered} in Theoretical Computer Science}, 3(2--3):93--263, 2009.

\bibitem{campbell1998inventory}
A.~Campbell, L.~Clarke, A.~Kleywegt, and M.~Savelsbergh.
\newblock The inventory routing problem.
\newblock In {\em Fleet management and logistics}, pages 95--113. Springer, 1998.

\bibitem{cheung2016submodular}
M.~Cheung, A.~N. Elmachtoub, R.~Levi, and D.~B. Shmoys.
\newblock The submodular joint replenishment problem.
\newblock {\em Mathematical Programming}, 158:207--233, 2016.

\bibitem{chrobak2014online}
M.~Chrobak.
\newblock Online aggregation problems.
\newblock {\em ACM SIGACT News}, 45(1):91--102, 2014.

\bibitem{dooly1998tcp}
D.~R. Dooly, S.~A. Goldman, and S.~D. Scott.
\newblock Tcp dynamic acknowledgment delay (extended abstract) theory and practice.
\newblock In {\em Proceedings of the thirtieth annual ACM symposium on Theory of computing}, pages 389--398, 1998.

\bibitem{dooly2001line}
D.~R. Dooly, S.~A. Goldman, and S.~D. Scott.
\newblock On-line analysis of the tcp acknowledgment delay problem.
\newblock {\em Journal of the ACM (JACM)}, 48(2):243--273, 2001.

\bibitem{fakcharoenphol2003tight}
J.~Fakcharoenphol, S.~Rao, and K.~Talwar.
\newblock A tight bound on approximating arbitrary metrics by tree metrics.
\newblock In {\em Proceedings of the thirty-fifth annual ACM symposium on Theory of computing}, pages 448--455, 2003.

\bibitem{federgruen1995analytical}
A.~Federgruen and D.~Simchi-Levi.
\newblock Analytical analysis of vehicle routing and inventory management problems. mo ball, tl magnanti, cl monma, gl nemhauser, eds. network routing. handbooks in or and ms, vol. 8, 1995.

\bibitem{fukunaga2014deliver}
T.~Fukunaga, A.~Nikzad, and R.~Ravi.
\newblock Deliver or hold: approximation algorithms for the periodic inventory routing problem.
\newblock In {\em Approximation, Randomization, and Combinatorial Optimization. Algorithms and Techniques (APPROX/RANDOM 2014)}. Schloss Dagstuhl-Leibniz-Zentrum fuer Informatik, 2014.

\bibitem{GW95}
M.~X. Goemans and D.~P. Williamson.
\newblock {A General Approximation Technique for Constrained Forest Problems}.
\newblock {\em SIAM J. Comput.}, 24(2):296--317, 1995.

\bibitem{karlin2001dynamic}
A.~R. Karlin, C.~Kenyon, and D.~Randall.
\newblock Dynamic tcp acknowledgement and other stories about e/(e-1).
\newblock In {\em Proceedings of the thirty-third annual ACM symposium on Theory of computing}, pages 502--509, 2001.

\bibitem{khouja2008review}
M.~Khouja and S.~Goyal.
\newblock A review of the joint replenishment problem literature: 1989--2005.
\newblock {\em European journal of operational Research}, 186(1):1--16, 2008.

\bibitem{levi2008constant}
R.~Levi, R.~Roundy, D.~Shmoys, and M.~Sviridenko.
\newblock A constant approximation algorithm for the one-warehouse multiretailer problem.
\newblock {\em Management Science}, 54(4):763--776, 2008.

\bibitem{LRS06}
R.~Levi, R.~Roundy, and D.~B. Shmoys.
\newblock Primal-dual algorithms for deterministic inventory problems.
\newblock {\em Math. Oper. Res.}, 31(2):267--284, 2006.

\bibitem{levi2}
R.~Levi, R.~Roundy, D.~B. Shmoys, and M.~Sviridenko.
\newblock A constant approximation algorithm for the one-warehouse multiretailer problem.
\newblock {\em Management Science}, 54(4):763--776, 2008.

\bibitem{levi2006improved}
R.~Levi and M.~Sviridenko.
\newblock Improved approximation algorithm for the one-warehouse multi-retailer problem.
\newblock In {\em International Workshop on Approximation Algorithms for Combinatorial Optimization}, pages 188--199. Springer, 2006.

\bibitem{mari2024online}
M.~Mari, M.~Paw{\l}owski, R.~Ren, and P.~Sankowski.
\newblock Online multi-level aggregation with delays and stochastic arrivals.
\newblock {\em arXiv preprint arXiv:2404.09711}, 2024.

\bibitem{mcmahan2021d}
J.~McMahan.
\newblock A $ d $-competitive algorithm for the multilevel aggregation problem with deadlines.
\newblock {\em arXiv preprint arXiv:2108.04422}, 2021.

\bibitem{nagarajan2016approximation}
V.~Nagarajan and C.~Shi.
\newblock Approximation algorithms for inventory problems with submodular or routing costs.
\newblock {\em Mathematical Programming}, 160:225--244, 2016.

\bibitem{nonner2009approximating}
T.~Nonner and A.~Souza.
\newblock Approximating the joint replenishment problem with deadlines.
\newblock {\em Discrete Mathematics, Algorithms and Applications}, 1(02):153--173, 2009.

\bibitem{PengWangWang}
L.~Peng, L.~Wang, and S.~Wang.
\newblock A review of the joint replenishment problem from 2006 to 2022.
\newblock {\em Management System Engineering}, 1, 11 2022.

\bibitem{Pinedo_2012}
M.~L. Pinedo.
\newblock {\em Scheduling: Theory, Algorithms, and Systems}.
\newblock Springer US, 2012.

\bibitem{sarmiento1999review}
A.~M. Sarmiento and R.~Nagi.
\newblock A review of integrated analysis of production-distribution systems.
\newblock {\em IIE transactions}, 31(11):1061--1074, 1999.

\bibitem{seiden2000guessing}
S.~S. Seiden.
\newblock A guessing game and randomized online algorithms.
\newblock In {\em Proceedings of the thirty-second annual ACM symposium on Theory of computing}, pages 592--601, 2000.

\bibitem{toth2002vehicle}
P.~Toth and D.~Vigo.
\newblock {\em The vehicle routing problem}.
\newblock SIAM, 2002.

\bibitem{touitou2023improved}
N.~Touitou.
\newblock Improved and deterministic online service with deadlines or delay.
\newblock {\em arXiv preprint arXiv:2306.05744}, 2023.

\end{thebibliography}

\newpage
\appendix
\section{Omitted Proofs}
\label{sec:appendix proofs}

\subsection{Single-Item Case}
\begin{proof}[Proof of Lemma~\ref{lem:single-item dual objective}]
    For each service $S_i$ triggered by Algorithm~\ref{alg:single-item} at time $t_i$, the sum of the backlog costs of the requests that trigger this service is exactly $s$ at time $t_i$, i.e., $\sum_{\rho \in B_i} b \cdot (t_i-d_\rho) = s$.
    
    In Case 1, we have
    \[ 
    \sum_{\rho \in B_i \cup H_{i-1}} \alpha_\rho = 
    \sum_{\rho \in B_i} \alpha_\rho = \sum_{\rho \in B_i} b \cdot (t_i - d_\rho) = s.
    \]
    In Case 2, we have
    \[\sum_{\rho \in H_{i-1}\cup B_i} \alpha_\rho = \sum_{\rho \in H_{i-1}} h \cdot (d_\rho - t_{i-1}) + \frac{s-h_{\text{sum}}}{s} \cdot  \sum_{\rho \in B_i} b \cdot (t_i-d_\rho) = h_{\text{sum}} + \frac{s-h_{\text{sum}}}{s} \cdot s = s.\]
\end{proof}

Before proving Lemma~\ref{lem:single-item feasibility}, we restate the dual solution presented in Section~\ref{sec:single-item-analysis}.

For each $i=1,\ldots,N$ do the following. 
Let $h_{\text{max}}$ be the holding cost of the request in $H_{i-1}$ which has the latest deadline at time $t_{i-1}$, and $b_{\text{max}}$ be the backlog cost of the request in $B^1_i$ which has the earliest deadline at time $t_i$ (see Figure~\ref{fig:local charging}). 
Set $h_{\text{max}}:=0$ if $H_{i-1} = \emptyset$ and $b_{\text{max}}:=0$ if $B_i^1 = \emptyset$.
For ease of notation, set $H_0=\emptyset$.
There are two cases
    \begin{itemize}
        \item \textbf{Case 1.} $h_{\text{max}} > b_{\text{max}}$: set $\alpha_\rho = b \cdot (t_i-d_\rho)$ for each $\rho \in B_i$, and set $\alpha_{\rho}=0$ for each $\rho \in H_{i-1}$. 
        \item \textbf{Case 2.} $h_{\text{max}} \leq b_{\text{max}}$: set $\alpha_\rho=h \cdot (d_\rho - t_{i-1})$ for each $\rho \in H_{i-1}$, and set $\alpha_\rho=\frac{s- h_{\text{sum}}}{s} \cdot b \cdot (t_i-d_\rho)$ for each $\rho \in B_i$, where $h_{\text{sum}}:=\sum_{\rho \in H_{i-1}} h\cdot (d_\rho-t_{i-1})$ is the sum of the holding costs of requests in $H_{i-1}$ at time $t_{i-1}$.
    \end{itemize}
For each $\rho$, set $\beta_{\rho,t}$ as follows (see Figure~\ref{fig:beta}):
\[
\beta_{\rho,t}= \begin{cases}
\max(0,\alpha_\rho - h \cdot (d_\rho - t)) &\text{for } a_\rho \leq t \leq d_\rho \\
\max(0,\alpha_\rho - b \cdot (t-d_\rho)) &\text{for } t > d_\rho.
\end{cases}
\]

\begin{proof}[Proof of Lemma~\ref{lem:single-item feasibility}]
    The way we set the beta variables ensures that Constrains~\eqref{const:holding}, \eqref{const:backlog}, and \eqref{const:dual nonnegativity} are satisfied. We only have to check Constraint~\eqref{const:beta cap}. 
    It is easy to see that for each service $S_i$, the $\beta_{\rho,t}$ variables associated to it sum up to at most $s$ at each time $t$, i.e., $\sum_{\rho \in H_{i-1} \cup B_i} \beta_{\rho, t} \leq s$ for each time $t$. 
    This is because by Lemma~\ref{lem:single-item dual objective} we have $\sum_{\rho \in H_{i-1} \cup B_i} \alpha_\rho = s$, and we know that $\beta_{\rho,t} \leq \alpha_\rho$. 
    It suffices to show that the $\beta$ variables for different services do not ``overlap''. 
    To this end, we show that for each service $S_i$, the $\beta_{\rho, t}$ variables for each $\rho \in H_{i-1} \cup B_i$ can be nonzero only for $t \in (t_{i-1},t_i)$, where $t_0=0$. 
    
    Note that we always have $\alpha_\rho \leq b\cdot (t_i - d_\rho)$ for each $\rho \in B_{i}$, 
    which ensures that $\beta_{\rho,t}=0$ for each $\rho \in B_i$ and $t \geq t_i$. 
    Also, $\beta_{\rho,t}=0$ for each $\rho \in H_{i-1}$ and $t\leq t_{i-1}$, as $\alpha_\rho \leq h \cdot (d_\rho - t_{i-1})$ for each $\rho \in H_{i-1}$. 
    Also note that by definition, the arrival time of each request $\rho \in B^2_i$ is after $t_{i-1}$, so $\beta_{\rho, t}=0$ for each $\rho \in B^2_i$ and $t \leq t_{i-1}$. 
    Thus, {\bf it suffices to show that  $\beta_{\rho, t}=0$ either if $\rho \in B^1_i$ and $t \leq t_{i-1}$ or if $\rho \in H_{i-1}$ and $t \geq t_i$} (we assume the undefined betas are 0 too).
    
    We first analyze the general case, where  $H_{i-1}$ and $B_i^1$ are non-empty. 
    We will address the special cases in the end.
    
    Let $\rho^*_H$ be the request with the latest deadline in $H_{i-1}$ and $\rho^*_B$ be the request with the earliest deadline in $B^1_i$.
    Since each service aggregates holding costs in increasing order of deadlines and all the requests in $B^1_i$ were alive at time $t_{i-1}$, we know that $d_{\rho^*_H} \leq d_{\rho^*_B}$ (see Figure~\ref{fig:local charging}). 
    Now we compare $h_{\text{max}} = h \cdot (d_{\rho^*_H}-t_{i-1})$ and $b_{\text{max}}=b \cdot (t_i-d_{\rho^*_B})$:
    \begin{itemize}
        \item If $h_{\text{max}} > b_{\text{max}}$, we set $\alpha_\rho=0$ for each $\rho \in H_{i-1}$, which means that $\beta_{\rho,t}=0$ for each $\rho \in H_{i-1}$ and each $t\geq t_i$.
        For each $\rho \in B_i$, we set $\alpha_\rho = b \cdot (t_i-d_\rho)$. For each time $a_\rho \leq t \leq t_{i-1}$ and each request $\rho \in B^1_i$ we have
        \begin{align*}
            \beta_{\rho, t} &= \max(0, \alpha_\rho - h\cdot (d_\rho - t)) \\
            &\leq \max(0, \alpha_{\rho^*_B} - h\cdot (d_{\rho^*_B} - t_{i-1})) \\
            &= \max(0, b_{\text{max}} - h\cdot (d_{\rho^*_B} - t_{i-1})) \\
            &\leq \max(0, h_{\text{max}} - h\cdot (d_{\rho^*_B} - t_{i-1})) \\
            &\leq \max(0, h_{\text{max}} - h\cdot (d_{\rho^*_H} - t_{i-1})) \\
            &=\max(0,0) \\
            &=0.
        \end{align*}
        \item If $h_{\text{max}} \leq b_{\text{max}}$, we set $\alpha_\rho=h \cdot (d_\rho - t_{i-1})$ for each $\rho \in H_{i-1}$. For each time $t \geq t_i$ we have
        \begin{align*}
            \beta_{\rho,t} &= \max(0, \alpha_\rho - b\cdot (t-d_\rho)) \\
            &\leq \max(0, \alpha_{\rho^*_H} - b\cdot (t_i-d_{\rho^*_H})) \\
            &\leq \max(0, h_{\text{max}} - b\cdot (t_i-d_{\rho^*_H})) \\
            &\leq \max(0, b_{\text{max}} - b\cdot (t_i-d_{\rho^*_H})) \\
            &\leq \max(0, b_{\text{max}} - b\cdot (t_i-d_{\rho^*_B})) \\
            &\leq \max(0,0) \\
            &= 0.
        \end{align*}
        Moreover, for each $\rho \in B_i$, we set $\alpha_\rho=\frac{s- h_{\text{sum}}}{s} \cdot b \cdot (t_i-d_\rho)$, where $h_{\text{sum}}:=\sum_{\rho \in H_{i-1}} h\cdot (d_\rho-t_{i-1})$. 
        Note that since the sum of the backlog costs of the requests in $B_i$ is exactly $s$ at time $t_i$, it follows that $b\cdot (t_i-d_\rho) \leq s$ for each individual $\rho \in B_i$. 
        Thus
        $\alpha_\rho = \frac{s- h_{\text{sum}}}{s} \cdot b \cdot (t_i-d_\rho) \leq s-h_{\text{sum}}$. 
        Each request $\rho \in B^1_i$ was alive at time $t_{i-1}$ but Algorithm~\ref{alg:single-item} did not serve it in the Holding Phase of service $S_{i-1}$. 
        This means that the holding cost of $\rho$ at time $t_{i-1}$ is more than $s-h_{\text{sum}}$, i.e., $h \cdot (d_\rho-t_{i-1}) > s-h_{\text{sum}}$. Therefore for each $\rho \in B^1_i$ and $a_\rho \leq t \leq t_{i-1}$ we have      
        \[\beta_{\rho,t} = \max(0, \alpha_\rho - h\cdot (d_\rho-t)) \leq \max(0, s-h_{\text{sum}} - h\cdot (d_\rho-t_{i-1})) = 0.\]
    \end{itemize} 
    Now we address the special cases. 

    If $H_{i-1} = \emptyset$, we set $h_{\text{max}}=0$, which means that Case 2 happens in the dual solution. 
    We only need to show $\beta_{\rho,t}=0$ for each $\rho \in B_i^1$ and $t\leq t_{i-1}$.
    If $i=1$, we have $a_\rho \geq t_{i-1}=0$, which means that $\beta_{\rho,t}=0$ for $t \leq t_{i-1} \leq a_\rho$.
    Otherwise, the proof follows from the exact argument used in the second case above with $h_{\text{sum}}=0$.

    If $H_{i-1}\neq \emptyset$ and $B_i^1=\emptyset$, Case 1 happens in the dual solution.
    Since $B_i^1$ is empty, we only need to check if $\beta_{\rho,t}=0$ for $\rho \in H_{i-1}$ and $t \geq t_i$, which is true as in Case 1 we set $\alpha_\rho=0$ for each $\rho \in H_{i-1}$.
\end{proof}

\subsection{Multi-Item Case}
\lemlocalcharging*

\begin{proof}
    Initially, we set all the $\alpha$ and $\beta$ variables to 0. 
    For each $\rho \in H$, we have $d_\rho \geq t_\ell$, and for each $\rho \in B$ we have $d_\rho \leq t^*$.
    Let $h_{\text{sum}}:=\sum_{\rho \in H} h \cdot (d_\rho-t_\ell)$ be the sum of the holding costs of requests in $H$ at time $t_\ell$. 
    Note that the design of Local Holding Phase of Algorithm~\ref{alg:multiple-item} for item $v$ implies that $h_{\text{sum}} \leq c(v)$. 
    Similarly, let $b_{\text{sum}}:=\sum_{\rho \in B} b \cdot (t^*-d_\rho)$ be the sum of the backlog costs of the request in $B$ at time $t^*$. 
    We show that $b_{\text{sum}} \geq c(v)$.
    This, by definition of $L_i^v$, holds in case 1. 
    In case 2, the definition of $\text{mature}_{t_i}(v)$ in Premature Backlog Phase implies that $\sum_{\rho \in B_{i,2}^v} b \cdot (\text{mature}_{t_i}(v)-d_\rho) = c(v)$. 

    Let $\rho^*$ be the request with the latest arrival time in $B$.
    Let $B^1$ be the set of requests in $B$ that have arrived before $t_\ell$, i.e., $B^1:=\{\rho \in B: a_\rho \leq t_\ell \}$. Let $B^2:= B \backslash B^1$.
    Let $\rho_H^*$ be the request with the latest deadline in $H$, and let $h_{\text{max}}=h \cdot (d_{\rho_H^*}-t_\ell)$ be its holding cost at time $t_\ell$.
    If $H=\emptyset$, set $h_{\text{max}}=0$.
    Similarly, let $\rho_B^*$ be the request with the earliest deadline in $B^1$, and $b_{\text{max}}=b \cdot (t^*-d_{\rho_B^*})$ be its backlog cost at time $t^*$. 
    Set $b_{\text{max}}=0$ if $B^1=\emptyset$.
    See Figure~\ref{fig:local charging} for a visualization, changing $(t_{i-1},t_i,H_{i-1},B_i^1,B_i^2)$ to $(t_\ell,t^*,H,B^1,B^2)$.
    Here is how we set $\alpha_\rho$ for each $\rho \in H \cup B$.
    There are two cases
    \begin{itemize}
        \item {\boldmath$h_{\text{max}} > 
        b_{\text{max}}$}: Set 

        \[
        \alpha_\rho:= \begin{cases}
            b \cdot (t^*-d_\rho) &\text{for } \rho \in B \backslash \{\rho^*\} \\
            c(v) - \sum_{\rho \in B \backslash \{\rho^*\}} b \cdot (t^*-d_\rho) &\text{for } \rho = \rho^*\\
            0 &\text{for } \rho \in H.
        \end{cases}
        \]
        It is clear that $\sum_{\rho \in B \cup H} \alpha_\rho = c(v)$. 
        Also, note that for both $B=L_i^v$ and $B=B_{i,2}^v$ we have ${\sum_{\rho \in B \backslash \{\rho^*\}} b \cdot (t^*-d_\rho) \leq c(v)}$, which means all the $\alpha_\rho$ variables defined are nonnegative.
        \item {\boldmath$h_{\text{max}} \leq b_{\text{max}}$}: Set $\alpha_\rho=h \cdot (d_\rho - t_\ell)$ for each $\rho \in H$. 
        Let $x=c(v) - h_{\text{sum}}$.
        Go over the requests $\rho \in B$ in increasing order of arrival time, and for each $\rho$, set ${\alpha_\rho=\min(x, b \cdot (t^*-d_\rho))}$, and decrease $x$ by $\alpha_\rho$.
        Note that since $h_{\text{sum}} \leq c(v)$, initially $0 \leq x \leq c(v)$. 
        Also, since ${\sum_{\rho \in B} b \cdot (t^*-d_\rho) = b_{\text{sum}} \geq c(v)}$, in the end, $x=0$, meaning that its initial value is distributed among $\alpha_\rho$ variables for $\rho \in B$. 
        So 
        \[{\sum_{\rho \in H\cup B} \alpha_\rho= \sum_{\rho \in H} h \cdot (d_\rho-t_\ell) + x = h_{\text{sum}} + (c(v) - h_{\text{sum}}) = c(v)}.\]
        Also, note that all the $\alpha$ variables are nonnegative.
        Since we go over the requests in $B$ in increasing order of arrival time and initial value of $x$ is at most $c(v)$, by the time we get to $\rho^*$, we have ${x = \max(0,c(v) - h_{\text{sum}} -  \sum_{\rho \in B \backslash \{\rho^*\}} b \cdot (t^*-d_\rho))}$.
        Since ${\sum_{\rho \in B \backslash \{\rho^*\}} b \cdot (t^*-d_\rho) \leq c(v)}$, we have
        \[
        \alpha_{\rho^*} \leq x \leq c(v) - \sum_{\rho \in B \backslash \{\rho^*\}} b \cdot (t^*-d_\rho).\]
    \end{itemize}
We have so far verified properties~\ref{cond:1} and~\ref{cond:5} in the lemma.
For each $\rho \in B \cup H$, set $\beta_{\rho,t}$ as follows:
\[
\beta_{\rho,t}= \begin{cases}
\alpha_\rho &\text{for } \max(a_\rho, d_\rho-\frac{\alpha_\rho}{h}) < t < d_\rho + \frac{\alpha_\rho}{b} \\
0 &\text{otherwise } .
\end{cases}
\]
From the definition of $\beta_{\rho,t}$ variables it is clear that $\beta_{\rho,t} \leq \alpha_\rho$ and constraints~\eqref{const JRP_D:holding}, and ~\eqref{const JRP_D:backlog} are satisfied. Also since $\alpha_\rho \geq 0$ for all $\rho \in B \cup H$, constraint~\eqref{const JRP_D:beta nonnegativity} is satisfied as well. 
So it remains to prove property~\ref{cond:2} in the lemma, which is to show that $\beta_{\rho,t}=0$ for each $\rho \in B \cup H$ and $t \notin (t_\ell,t^*)$.

So for each $\rho \in B \cup H$, it suffices to verify that $(\max(a_\rho, d_\rho-\frac{\alpha_\rho}{h}), d_\rho + \frac{\alpha_\rho}{b}) \subseteq (t_\ell,t^*)$, which is equivalent to the following two conditions:
\begin{enumerate}
    [label*=\textbf{C\arabic*}]
    \item\label{C1-local} $a_\rho \geq t_\ell$ or $\alpha_\rho \leq h \cdot (d_\rho - t_\ell)$ 
    \item\label{C2-local}$\alpha_\rho \leq b \cdot (t^*-d_\rho)$.
\end{enumerate}

\ref{C1-local} holds for each $\rho \in H$ because $\alpha_\rho \leq h \cdot (d_\rho - t_\ell)$. 
Also,~\ref{C2-local} holds for each $\rho \in B$. 
Note that in the case where $h_{\max}>b_{\max}$ and $\rho = \rho^*$, since $\sum_{\rho \in B}b \cdot (t^*-d_\rho)=b_{\text{sum}}\geq c(v)$, it follows that $\alpha_{\rho^*}=c(v)-\sum_{\rho \in B \backslash \{\rho^*\}} b \cdot (t^*-d_\rho) \leq b \cdot (t^*-d_{\rho^*})$. 
By definition, the arrival time of each request $\rho \in B^2$ is after $t_\ell$, so~\ref{C1-local} holds for it.

Thus, {\bf it suffices to show that~\ref{C2-local} holds for each $\rho \in H$ and~\ref{C1-local} holds for each $\rho \in B^1$}. 
We first analyze the general case, where $H$ and $B^1$ are non-empty. 
We will address the special cases in the end.

In Local Holding Phase of service $S_\ell$ for item $v$ in Algorithm~\ref{alg:multiple-item}, we go over the active requests for $v$ in increasing order of deadlines.
All the requests in $B^1$ were active at time $t_\ell$ but were not served during this phase.
This means that the last request in $H=H_{\ell,1}^v$ has an earlier deadline than the first request in $B^1$, i.e., $d_{\rho^*_H} \leq d_{\rho^*_B}$. 
Now we compare $h_{\text{max}} = h \cdot (d_{\rho^*_H}-t_\ell)$ and $b_{\text{max}}=b \cdot (t^*-d_{\rho^*_B})$:
\begin{itemize}
    \item Case 1: $h_{\text{max}} > b_{\text{max}}$. We set
    $\alpha_\rho=0$ for each $\rho \in H$, so~\ref{C2-local} holds for it.
    For each $\rho \in B^1$, we have
    \begin{align*}
        \alpha_\rho &\leq b \cdot (t^*-d_\rho)\\
        &\leq b \cdot (t^*-d_{\rho^*_B})\\
        &= b_{\text{max}}\\
        &< h_{\text{max}}\\
        &=h \cdot (d_{\rho^*_H}-t_\ell)\\
        &\leq h \cdot (d_{\rho^*_B}-t_\ell)\\
        &\leq h \cdot (d_\rho - t_\ell), 
    \end{align*}
    which means that~\ref{C1-local} holds for $\rho$.
    
    \item Case 2: $h_{\text{max}} \leq b_{\text{max}}$. We set $\alpha_\rho=h \cdot (d_\rho - t_\ell)$ for each $\rho \in H$. Therefore
    \begin{align*}
        \alpha_\rho &= h \cdot (d_\rho - t_\ell) \\
        &\leq h \cdot (d_{\rho^*_H} - t_\ell)\\
        &= h_{\text{max}}\\
        &\leq b_{\text{max}}\\
        &=b \cdot (t^*-d_{\rho^*_B})\\
        &\leq b \cdot (t^*-d_{\rho^*_H})\\
        &\leq b \cdot (t^*-d_\rho),
    \end{align*}
    which means that~\ref{C2-local} holds for $\rho$.
    
    Finally, for each $\rho \in B^1$, we have $\alpha_\rho \leq c(v)-h_{\text{sum}}$. 
    Each request $\rho \in B^1$ was active at time $t_\ell$ but Algorithm~\ref{alg:multiple-item} did not satisfy it in Local Holding Phase for item $v$ in service $S_i$ because it ran out of budget. 
    This means that the holding cost of $\rho$ at time $t_\ell$ is more than $c(v)-h_{\text{sum}}$, i.e., $h \cdot (d_\rho-t_\ell) > c(v)-h_{\text{sum}}$. 
    Therefore for each $\rho \in B^1$ 
    \[ \alpha_\rho \leq c(v)-h_{\text{sum}} < h \cdot (d_\rho-t_\ell), \]
    which implies that~\ref{C1-local} holds for $\rho$.
\end{itemize}

Now we address the special cases. 
If $H = \emptyset$, we set $h_{\text{max}}=0$, which means that $h_{\text{max}} \leq b_{\text{max}}$.
We do not have any requests in $H$, which means that we only need to show~\ref{C1-local} holds for each $\rho \in B^1$.
If $S_\ell$ is the first service that includes $v$, we have $t_{\ell}=0$, which means that $a_\rho\geq t_\ell$.
Otherwise, the proof follows from the exact argument used in the second case above with $h_{\text{sum}}=0$.

If $H\neq \emptyset$ and $B^1=\emptyset$, we have $h_{\text{max}} > b_{\text{max}}$.
Since $B^1$ is empty, we only need to check if~\ref{C2-local} holds for each $\rho \in H$.
But note that in this case, we set $\alpha_\rho=0$ for each $\rho \in H$, which implies that~\ref{C2-local} holds for $\rho$.
\end{proof}

\lemtwosidedglobalcharging*

\begin{proof}
Initially, we set all the $\alpha$, $\beta$, and $\gamma$ variables to 0.
Let
\[
SB:=\sum_{v \in V_{i,1}\cap V_{i-1}}\left( \sum_{\rho \in B_{i,1}^v} b \cdot (t_i-d_\rho)-c(v) \right)
\]
be the surplus backlog cost of the vertices in $V_{i,1} \cap V_{i-1}$ in Mature Backlog Phase of $S_i$.
Algorithm~\ref{alg:multiple-item} in Mature Backlog Phase waits until the sum of the surplus backlog costs of the mature items (all the items in $V_{i,1}$) reaches $c(r)$ and then triggers a service.
Thus, $SB \leq c(r)$.

For each $\rho \in H$, we have $d_\rho \geq t_{i-1}$, and for each $\rho \in B$ we have $d_\rho \leq t_i$.
Let $b_{\text{sum}}:=\sum_{\rho \in B} b \cdot (t_i-d_\rho)$ be the sum of the backlog costs of the requests in $B$ at time $t_i$.
For each $v \in V_{i,1}$ we have $B_{i,1}^v=L_i^v \cup R_i^v$, where by definition, $\sum_{\rho \in L_i^v} b \cdot (t_i-d_\rho) \leq c(v)$.
This means that for each $v \in V_{i,1} \cap V_{i-1}$, we have 
\[\sum_{\rho \in R_i^v} b \cdot (t_i-d_\rho) \geq \sum_{\rho \in B_{i,1}^v} b \cdot (t_i-d_\rho) - c(v).\]
Thus, 
\[b_{\text{sum}} = \sum_{\rho \in B} b \cdot (t_i-d_\rho)
=\sum_{v \in V_{i,1}\cap V_{i-1}}\sum_{\rho \in R_i^v} b \cdot (t_i-d_\rho)
\geq \sum_{v \in V_{i,1}\cap V_{i-1}}\left( \sum_{\rho \in B_{i,1}^v} b \cdot (t_i-d_\rho)-c(v) \right) = SB.\]
Let $h_{\text{sum}}:=\sum_{\rho \in H} h \cdot (d_\rho-t_{i-1})$ be the sum of the holding costs of requests in $H$ at time $t_{i-1}$. 
By the design of Global Holding Phase of Algorithm~\ref{alg:multiple-item} we have $h_{\text{sum}} \leq c(r)$. 

Let $B^1$ be the set of requests in $B$ that have arrived before $t_{i-1}$, i.e., $B^1:=\{\rho \in B: a_\rho \leq t_{i-1} \}$. 
Let $B^2:= B \backslash B^1$.
Let $\rho_H^*$ be the request with the latest deadline in $H$, and let $h_{\text{max}}=h \cdot (d_{\rho_H^*}-t_{i-1})$ be its holding cost at time $t_{i-1}$.
If $H=\emptyset$, set $h_{\text{max}}=0$.
Similarly, let $\rho_B^*$ be the request with the earliest deadline in $B^1$, and $b_{\text{max}}=b \cdot (t_i-d_{\rho_B^*})$ be its backlog cost at time $t_i$. 
Set $b_{\text{max}}=0$ if $B^1=\emptyset$.
See Figure~\ref{fig:local charging} for a visualization, changing $(H_{i-1},B_i^1,B_i^2)$ to $(H,B^1,B^2)$.

Here is how we set $\alpha_\rho$ for each $\rho \in H \cup B$.
    There are two cases
    \begin{itemize}
        \item {\boldmath$h_{\text{max}} > 
        b_{\text{max}}$}: Set 

        \[
        \alpha_\rho:= \begin{cases}
            b \cdot (t_i-d_\rho) \cdot \frac{SB}{b_{\text{sum}}} &\text{for } \rho \in B \\
            0 &\text{for } \rho \in H.
        \end{cases}
        \]
        It is clear that $\alpha_\rho \geq 0$ for each $\rho \in B \cup H$, and 
        
        \[\sum_{\rho \in B \cup H} \alpha_\rho= \sum_{\rho \in B} b \cdot (t_i-d_\rho) \cdot \frac{SB}{b_{\text{sum}}} = b_{\text{sum}} \cdot \frac{SB}{b_{\text{sum}}}= SB.\] 
        
        \item {\boldmath$h_{\text{max}} \leq b_{\text{max}}$}: Set 

        \[
        \alpha_\rho:= \begin{cases}
            b \cdot (t_i-d_\rho) \cdot \frac{SB-h_{\text{sum}}}{b_{\text{sum}}} &\text{for } \rho \in B 
             \text{ if } h_{\text{sum}}<SB\\
            0 &\text{for } \rho \in B \text{ if } h_{\text{sum}} \geq SB\\
            h \cdot (d_\rho - t_{i-1}) &\text{for } \rho \in H.
        \end{cases}
        \]

        Note that $\alpha_\rho \geq 0$ for each $\rho \in B \cup H$.
        If $h_{\text{sum}} \geq SB$
        we have
        \[
        \sum_{\rho \in H \cup B}\alpha_\rho = \sum_{\rho \in H} \alpha_\rho = \sum_{\rho \in H} h \cdot (d_\rho-t_{i-1}) =h_{\text{sum}} \geq SB.
        \]
        If $h_{\text{sum}} < SB$, we have
        \[
        \sum_{\rho \in H \cup B} \alpha_\rho = \sum_{\rho \in H} h\cdot (d_\rho-t_{i-1}) + \sum_{\rho \in B} b \cdot (t_i-d_\rho) \cdot \frac{SB-h_{\text{sum}}}{b_{\text{sum}}}
        = h_{\text{sum}} + b_{\text{sum}} \cdot \frac{SB-h_{\text{sum}}}{b_{\text{sum}}}=SB.
        \]
    \end{itemize}
Thus, in both cases, we showed that 
\[
\sum_{\rho \in H \cup B} \alpha_\rho \geq SB = \sum_{v \in V_{i,1}\cap V_{i-1}}\left( \sum_{\rho \in B_{i,1}^v} b \cdot (t_i-d_\rho)-c(v) \right),
\]
which verifies property~\ref{cond two-sided:1} in the lemma.

Note that in both cases above we have
\[
 \sum_{\rho \in H \cup B} \alpha_\rho \leq \max(SB,h_{\text{sum}}) \leq c(r).
\]
Since $SB \leq b_{\text{sum}}$, we have $\alpha_\rho \leq b \cdot (t_i-d_\rho)$ for each $\rho \in B$.
Also, for each $\rho \in H$ we have $\alpha_\rho \leq h \cdot (d_\rho-t_{i-1})$.

For each $\rho \in B \cup H$, set $\beta_{\rho,t}$ as follows:
\[
\beta_{\rho,t}= \begin{cases}
\alpha_\rho &\text{for } \max(a_\rho, d_\rho-\frac{\alpha_\rho}{h}) < t < d_\rho + \frac{\alpha_\rho}{b} \\
0 &\text{otherwise } .
\end{cases}
\]
From the definition of $\beta_{\rho,t}$ variables it is clear that $\beta_{\rho,t} \leq \alpha_\rho$ (property~\ref{cond two-sided:3}) and constraints~\eqref{const JRP_D:holding}, and ~\eqref{const JRP_D:backlog} are satisfied. 
Also since $\alpha_\rho \geq 0$ for all $\rho \in B \cup H$, constraint~\eqref{const JRP_D:beta nonnegativity} is satisfied as well. 

For each time $t$ and item $v \in V_{i,1}\cap V_{i-1}$ set $\gamma_{v,t}:=\sum_{\rho:v_\rho=v} \beta_{\rho,t}$.
This way, Constraint~\eqref{const JRP_D:beta cap} and property~\ref{cond two-sided:5} in the lemma are satisfied.
Moreover, we have $\gamma_{v,t} \geq 0$ and 
\[
\sum_v \gamma_{v,t}
= \sum_{v} \sum_{\rho:v_\rho=v}\beta_{\rho,t} 
\leq 
\sum_{v} \sum_{\rho:v_\rho=v}\alpha_{\rho} = \sum_{\rho \in B \cup H} \alpha_\rho \leq c(r), 
\]
which means that Constraints~\eqref{const JRP_D:gamma nonnegativity} and~\eqref{const JRP_D:gamma cap} hold too.
Therefore, our proposed dual solution satisfies all the constraints in $[\text{JRP}_{\text{D}}]$, which shows that property~\ref{cond two-sided:4} in the lemma holds.

So it remains to prove property~\ref{cond two-sided:2} in the lemma.
To this end, we show that $\beta_{\rho,t}=0$ for each $\rho \in B \cup H$ and $t \notin (t_{i-1},t_i)$.
Note that since $\gamma_{v,t}=\sum_{\rho:v_\rho=v}\beta_{\rho,t}$, this proves that $\gamma_{v,t}=0$ for each $t \notin (t_{i-1},t_i)$.

So for each $\rho \in B \cup H$, it suffices to verify that $(\max(a_\rho, d_\rho-\frac{\alpha_\rho}{h}), d_\rho + \frac{\alpha_\rho}{b}) \subseteq (t_{i-1},t_i)$, which is equivalent to the following two conditions:
\begin{enumerate}[label*=\textbf{C\arabic*}]
    \item\label{C1-global} $a_\rho \geq t_{i-1}$ or $\alpha_\rho \leq h \cdot (d_\rho - t_{i-1})$ 
    \item\label{C2-global}$\alpha_\rho \leq b \cdot (t_i-d_\rho)$.
\end{enumerate}

\ref{C1-global} holds for each $\rho \in H$ because $\alpha_\rho \leq h \cdot (d_\rho - t_{i-1})$. 
Also,~\ref{C2-global} holds for each $\rho \in B$. 
By definition, the arrival time of each request $\rho \in B^2$ is after $t_{i-1}$, so~\ref{C1-global} holds for it.

Thus, {\bf it suffices to show that~\ref{C2-global} holds for each $\rho \in H$ and~\ref{C1-global} holds for each $\rho \in B^1$}. 
We first analyze the general case, where $H$ and $B^1$ are non-empty. 
We will address the special cases in the end.

In Global Holding Phase of service $S_{i-1}$ in Algorithm~\ref{alg:multiple-item}, we go over the active requests for items $V_{i-1}$ in increasing order of deadlines.
We know $B^1 \subseteq B = \bigcup_{v \in V_{i,1}\cap V_{i-1}} R_i^v$, which shows that all the requests in $B^1$ are for items in $V_{i-1}$.
Moreover, all the requests in $B^1$ were active at time $t_{i-1}$, which means they were among the candidates in Global Holding Phase of $S_{i-1}$, but they were not served during this phase.
This means that the last request in $H=H_{i-1,2}$ has an earlier deadline than the first request in $B^1$, i.e., $d_{\rho^*_H} \leq d_{\rho^*_B}$. 
Now we compare $h_{\text{max}} = h \cdot (d_{\rho^*_H}-t_{i-1})$ and $b_{\text{max}}=b \cdot (t_i-d_{\rho^*_B})$:
\begin{itemize}
    \item Case 1: $h_{\text{max}} > b_{\text{max}}$. We set
    $\alpha_\rho=0$ for each $\rho \in H$, so~\ref{C2-global} holds for it.
    For each $\rho \in B^1$, we have
    \begin{align*}
        \alpha_\rho &\leq b \cdot (t_i-d_\rho)\\
        &\leq b \cdot (t_i-d_{\rho^*_B})\\
        &= b_{\text{max}}\\
        &< h_{\text{max}}\\
        &=h \cdot (d_{\rho^*_H}-t_{i-1})\\
        &\leq h \cdot (d_{\rho^*_B}-t_{i-1})\\
        &\leq h \cdot (d_\rho - t_{i-1}), 
    \end{align*}
    which means that~\ref{C1-global} holds for $\rho$.
    
    \item Case 2: $h_{\text{max}} \leq b_{\text{max}}$. We set $\alpha_\rho=h \cdot (d_\rho - t_{i-1})$ for each $\rho \in H$. Therefore
    \begin{align*}
        \alpha_\rho &= h \cdot (d_\rho - t_{i-1}) \\
        &\leq h \cdot (d_{\rho^*_H} - t_{i-1})\\
        &= h_{\text{max}}\\
        &\leq b_{\text{max}}\\
        &=b \cdot (t_i-d_{\rho^*_B})\\
        &\leq b \cdot (t_i-d_{\rho^*_H})\\
        &\leq b \cdot (t_i-d_\rho),
    \end{align*}
    which means that~\ref{C2-global} holds for $\rho$.
    
    Finally, for each $\rho \in B^1$, if $h_{\text{sum}}\geq SB$ we have $\alpha_\rho = 0 \leq c(r)-h_{\text{sum}}$, and if $h_{\text{sum}} < SB$
    we have 
    \[
    \alpha_\rho 
    \leq \sum_{\rho' \in B} \alpha_\rho
    = \sum_{\rho' \in B} 
    b \cdot (t_i-d_{\rho'}) \cdot \frac{SB-h_{\text{sum}}}{b_{\text{sum}}}
    =b_{\text{sum}} \cdot \frac{SB-h_{\text{sum}}}{b_{\text{sum}}} 
    \leq c(r)-h_{\text{sum}}.
    \]
     
    Each request $\rho \in B^1$ is for an item $v \in V_{i-1}$ and was active at time $t_{i-1}$ but Algorithm~\ref{alg:multiple-item} did not satisfy it in Global Holding Phase in service $S_{i-1}$ because it ran out of budget. 
    This means that the holding cost of $\rho$ at time $t_{i-1}$ was more than the remaining budget $c(r)-h_{\text{sum}}$, i.e., $h \cdot (d_\rho-t_{i-1}) > c(r)-h_{\text{sum}}$. 
    Therefore for each $\rho \in B^1$ 
    \[ \alpha_\rho \leq c(r)-h_{\text{sum}} < h \cdot (d_\rho-t_{i-1}), \]
    which implies that~\ref{C1-global} holds for $\rho$.
\end{itemize}

Now we address the special cases. 
If $H = \emptyset$, we set $h_{\text{max}}=0$, which means that $h_{\text{max}} \leq b_{\text{max}}$.
We do not have any requests in $H$, which means that we only need to show~\ref{C1-global} holds for each $\rho \in B^1$.
The proof follows from the exact argument used in the second case above with $h_{\text{sum}}=0$.

If $H\neq \emptyset$ and $B^1=\emptyset$, we have $h_{\text{max}} > b_{\text{max}}$.
Since $B^1$ is empty, we only need to check if~\ref{C2-global} holds for each $\rho \in H$.
But note that in this case, we set $\alpha_\rho=0$ for each $\rho \in H$, which implies that~\ref{C2-global} holds for $\rho$.
\end{proof}

\begin{proof}[Proof of Lemma~\ref{lem:charged-twice}]
    In each step of the dual assignment, a subset of requests is involved, and only the associated variables to these requests might be modified.
    Here is the breakdown of the requests involved in each step of Dual Assignment $i$:
    \begin{itemize}
        \item In step~\ref{D1}, $\LC(v,S_i)$ is called for each $v \in V_{i,1}$, which modifies the requests in $H_{\ell,1}^v \cup L_i^v$, where $\ell$ is the last service before $S_i$ that includes item $v$.
        \item In step~\ref{D21}, $\LC(v,S_{i-1})$ is invoked for each $v \in V_{i-1,2}$, which modifies the requests in $H_{\ell,1}^v \cup B_{i-1,2}^v$, where $\ell$ is the last service before $S_{i-1}$ that includes $v$.
        \item In step~\ref{D221}, for each $\rho \in R_i^1 = \bigcup_{v \in V_{i,1} \backslash V_{i-1}} R_i^v$, $\GC(\rho,S_i)$ is called, which only modifies the request $\rho$.
        \item In step~\ref{D222}, $\TSGC(S_i)$ is called, which uses the requests in $H_{i-1,2}\cup R_i^2$, where $R_{i,2}=\bigcup_{v \in V_{i,1}\cap V_{i-1}} R_i^v$.
    \end{itemize}
    Note that during Dual Assignment $i$, only the requests that are satisfied by the services $S_1,\ldots,S_i$ are used in the chargings, except for step~\ref{D21}, where the requests in $B_{i-1,2}^v$ might be used, which are not necessarily satisfied by the first $i$ services.
    This is because $B_{i,2}^v$ includes all the requests for $v$ that in Premature Backlog Phase of service $S_{i-1}$ contribute to making $v$ mature at time $\text{mature}_{t_{i-1}}(v)>t_{i-1}$.
    Some of these requests might have deadlines before $t_{i-1}$ and some might have deadlines after $t_{i-1}$. 
    In Premature Backlog Phase of $S_{i-1}$ only the overdue requests are served, i.e., the former set of requests.
    The latter set might remain unsatisfied after  $S_{i-1}$ is done.
    But the key observation is that step~\ref{D21} happens only if $t_i > \text{mature}_{t_{i-1}}(v_{k+1}) \geq \text{mature}_{t_{i-1}}(v)$, which means that service $S_i$ happens after the deadline of $\rho$.
    Thus, the next time request $\rho$ is involved in a charging for some service $S_\ell$, we have $d_\rho<t_\ell$, which implies that request $\rho$ is definitely satisfied by then.

    To prove the lemma, it suffices to show that after a request $\rho$ is satisfied, it is involved in at most one local charging and one global charging.
    We prove this by considering different phases of the algorithm at which request $\rho$ for item $v$ can be satisfied.
    \begin{itemize}
        \item If $\rho$ is satisfied during Mature Backlog Phase of $S_i$, there are two cases:
        \begin{itemize}
            \item $v \in L_i^v$, in which case it can only be used during $\LC(v,S_i)$ in step~\ref{D1} in Dual Assignment $i$.
            \item $v \in R_i^v$, in which case, if $v\notin V_{i-1}$, it can only be used during $\GC(\rho,S_i)$ in step~\ref{D221} of Dual Assignment $i$, and if $v \in V_{i-1}$, it can only be used during $\TSGC(S_i)$ in step~\ref{D222} of Dual Assignment $i$.
        \end{itemize}
        Note that $\rho$ can be in $L_i^v \cap R_i^v$, which means that it can be charged at most once during the local chargings and once during the global chargings.

        \item If $\rho$ is satisfied during Premature Backlog Phase of service $S_i$, it can only be used in $\LC(v,S_i)$ in step~\ref{D21} of Dual Assignment $i$.

        \item If $\rho$ is satisfied during Local Holding Phase of $S_i$, it can only be used during $\LC(v,S_\ell)$ in step~\ref{D1} of Dual Assignment $\ell$ or step~\ref{D21} of Dual Assignment $\ell+1$, where $S_\ell$ is the next service after $S_i$ that includes $v$.

        \item If $\rho$ is satisfied during Global Holding Phase of $S_i$, it can only be used in $\TSGC(S_{i+1})$ during step~\ref{D222} of Dual Assignment $i+1$.
    \end{itemize}
\end{proof}

\begin{proof}[Proof of Lemma~\ref{lem:local-beta}]
    Let $S_{i_1},\ldots,S_{i_m}$ be the services that include item $v$, which happen at times $t_{i_1},\ldots,t_{i_m}$.
    Each time $\LC(v,S_{i_j})$ is called for some $j=1,\ldots,m$, the $\beta^*$ values are calculated according to Lemma~\ref{lem:local charging}, and then $\beta_{\rho,t}$ is increased by $\frac{1}{2}\beta_{\rho,t}^*$.
    For each $j=1,\ldots,m$, $\LC(v,S_{i_j})$ is called at most once, and it is during one of the following two steps:
 
    \begin{itemize}
        \item In step~\ref{D1} of Dual Assignment $i_j$.
        In this case, since $v \in V_{i_j,1}$, in Lemma~\ref{lem:local charging}, $t^*$ is set to $t_{i_j}$, and Lemma~\ref{cond:2} guarantees that $\beta_{\rho,t}^*=0$ when $t \notin [t_{i_{j-1}},t_{i_j}]$.
        \item In step~\ref{D21} of Dual Assignment $i_{j}+1$.
        When this happens, $\text{mature}_{t_{i_j}}(v)<t_{i_j+1}$, which means that $t^*$ in Lemma~\ref{lem:local charging} is strictly less than $t_{i_j+1}$. Since $t_{i_j+1}\leq t_{i_{j+1}}$,  Lemma~\ref{cond:2} ensures that $\beta_{\rho,t}^*=0$ for all $t \notin [t_{i_{j-1}},t_{i_{j+1}})$.
    \end{itemize}
    Thus, $\LC(v,S_{i_{j}})$ can only change $\beta_{\rho,t}$ for $t \in [t_{i_{j-1}},t_{i_{j+1}})$.
    This shows that for each time $t \in [t_{i_{j}} ,t_{i_{j+1}})$, $\sum_{\rho:v_\rho=v} \beta_{\rho,t}$ can only increase during $\LC(v,S_{i_j})$ and $\LC(v,S_{i_{j}+1})$, each of which is invoked at most once.
\end{proof}

\begin{proof}[Proof of Lemma~\ref{lem:dual feasibility}]
    Initially, all the variables in [$\text{JRP}_\text{D}$] are set to 0, which means that all the constraints hold.
    We verify that all the constraints in [$\text{JRP}_\text{D}$] remain satisfied after the chargings. 
    \begin{itemize}
        \item  Constraint~\eqref{const JRP_D:beta cap}: 
        Fix an item $v$ and a time $t$.
        Lemma~\ref{lem:local-beta} shows that $\sum_{\rho:v_\rho=v} \beta_{\rho,t}$ can increase during at most two local chargings.
        Each time $\LC(v,S_i)$ is called for some service $S_i$, the $\beta^*$ values are calculated according to Lemma~\ref{lem:local charging}, and then $\frac{1}{4}\beta_{\rho,t}^*$ is added to $\beta_\rho$.
        We know from Lemma \ref{cond:1} and \ref{cond:3} that 
        $\sum_\rho \beta_{\rho,t}^* \leq \sum_\rho \alpha_\rho^*=c(v)$. This proves that considering only the local chargings, we have $\sum_{\rho: \ a_\rho \leq t, \ v_\rho=v} \beta_{\rho,t} - \gamma_{v,t}  \leq \frac{1}{2}c(v)$ for all times $t$, where $\gamma_{v,t}=0$ after the local chargings.

        Each time $\GC(\rho,S_i)$ is invoked for some request $\rho$ for item $v$, the values of $\beta_{\rho,t}$ and $\gamma_{v_\rho,t}$ increase by the same value, which ensures that Constraint~\eqref{const JRP_D:beta cap} remains satisfied.

        Each time $\TSGC(S_i)$ is called, the values $\beta^*$ and $\gamma^*$ are calculated according to Lemma~\ref{lem:two-sided global charge}, and then $\frac{1}{2}\beta_{\rho,t}^*$ and $\frac{1}{2}\gamma_{v,t}^*$ are added to $\beta_{\rho,t}$ and $\gamma_{v,t}$, respectively. 
        By Lemma~\ref{cond two-sided:5} we have that $\sum_{\rho:a_\rho \leq t, v_\rho = v}\beta_{\rho, t}^* - \gamma_{v,t}^*= 0$
        for each item $v$ involved and each time $t$.
        Therefore Constraint~\eqref{const JRP_D:beta cap} remains satisfied after each call to $\TSGC(S_i)$.

        Lastly, notice that in step~\ref{D223}, all the increases in $\beta$ and $\gamma$ variables during steps~\ref{D221} and~\ref{D222} are scaled down by the same factor, which means that Constraint~\eqref{const JRP_D:beta cap} still holds after this step.
        
        \item Constraint~\eqref{const JRP_D:gamma cap}: 
        The local chargings do not change the $\gamma$ variables.
        These variables can only change during step~\ref{D22}.
        Lemma~\ref{lem:global charge arrival} implies that the $\gamma_{v,t}$ variables that are increased during step~\ref{D221} for $S_i$ all have $t_{i-1}<t\leq t_i$. 
        Also Lemma~\ref{cond two-sided:2} shows that $\gamma_{v,t}$ variables that are increased in step~\ref{D222} of Dual Assignment $i$ all have $t_{i-1} < t < t_i$.
        This means for each time $t \in (t_{i-1},t_i]$, only Dual Assignment $i$ can change the values of $\gamma_{v,t}$.

        During step~\ref{D221} of Dual Assignment $i$, $\GC(\rho,S_i)$ is called, and for each time $t \in (t_{i-1},t_i]$, it increases $\sum_v \gamma_{v,t}$ by at most the amount it increases $\sum_\rho \alpha_\rho$.
        During step~\ref{D222} of Dual Assignment $i$, $\TSGC(S_i)$ is called, which computes $\alpha^*$, $\beta^*$, and $\gamma^*$ according to Lemma~\ref{lem:two-sided global charge}, and increases the values of $\alpha_\rho$, $\beta_{\rho,t}$, and $\gamma_{v,t}$ by $\frac{1}{2}\alpha^*_\rho$, $\frac{1}{2}\beta^*_{\rho,t}$, and $\frac{1}{2}\gamma^*_{v,t}$, respectively.
        By Lemma~\ref{cond two-sided:5} and~\ref{cond two-sided:3}, for each time $t$ we have $\sum_v \gamma_{v,t}^* \leq \sum_{\rho} \beta_{\rho,t}^*\leq \sum_\rho \alpha_\rho^*$.
        Thus, during steps~\ref{D221} and~\ref{D222}, the total amount of increase in $\sum_v \gamma_{v,t}$ is at most as much as the increase in $\sum_\rho \alpha_\rho$.
        Step~\ref{D223} ensures that the total increase in $\sum_\rho \alpha_\rho$ is at most $c(r)$ (by scaling down the increases), which means that the total increase of $\sum_v \gamma_{v,t}$ for each time $t$ during step $\ref{D22}$ of Dual Assignment $i$ is at most $c(r)$.
        Therefore, for each time $t$, Constraint~\eqref{const JRP_D:gamma cap} is satisfied.

        \item Constraints \eqref{const JRP_D:holding} and \eqref{const JRP_D:backlog}:
        Consider a request $\rho$ for item $v$ and a time $t$.
        Assume this request was never involved in a $\GC(\rho,S_i)$.
        By Lemma~\ref{lem:charged-twice}, this request can be modified in at most two calls to $\LC(.)$, and at most one call to $\TSGC(.)$, during which $\alpha_\rho$ and $\beta_{\rho,t}$ increase by $\frac{1}{4}\alpha_\rho^1+\frac{1}{4}\alpha_\rho^2+\nu\frac{1}{2}\alpha_\rho^3$ and  $\frac{1}{4}\beta_{\rho,t}^1+\frac{1}{4}\beta_{\rho,t}^2+\nu\frac{1}{2}\beta_{\rho,t}^3$, respectively, where $(\alpha_\rho^1,\beta_{\rho,t}^1)$ and $(\alpha_\rho^2,\beta_{\rho,t}^2)$ are obtained using Lemma~\ref{lem:local charging},  $(\alpha_\rho^3,\beta_{\rho,t}^3)$ is obtained using Lemma~\ref{lem:two-sided global charge}, and $\nu \in (0,1]$ is the coefficient used in step~\ref{D223} to scale down the increase that has happened during $\TSGC(.)$ invoked in step~\ref{D222}.
        By Lemma~\ref{cond:4} and~\ref{cond two-sided:4}, we know that $\alpha_\rho^\ell$ and $\beta_{\rho,t}^\ell$ satisfy Constraints~\eqref{const JRP_D:holding} and~\eqref{const JRP_D:backlog} for $\ell=1,2,3$. This means that if $a_\rho \leq t \leq d_\rho$ we have
        \[\alpha_\rho^\ell - \beta_{\rho, t}^\ell \leq h \cdot (d_\rho - t),\]
        and if $t>d_\rho$ we have
        \[\alpha_\rho^\ell - \beta_{\rho, t}^\ell \leq b \cdot (t - d_\rho),\]
        for $\ell=1,2,3$.
        Therefore, if $a_\rho \leq t \leq d_\rho$ we have
        \[
        \alpha_\rho-\beta_{\rho,t}
        =\frac{1}{4}(\alpha_\rho^1-\beta_{\rho,t}^1)+\frac{1}{4}(\alpha_\rho^2-\beta_{\rho,t}^2)+\nu\frac{1}{2}(\alpha_\rho^3 -\beta_{\rho,t}^3)
        \leq (\frac{1}{4}+\frac{1}{4}+\nu \frac{1}{2}) h \cdot (d_\rho-t)
        \leq h \cdot (d_\rho-t),
        \]
        and if $t>d_\rho$ we have
        \[
        \alpha_\rho-\beta_{\rho,t}
        =\frac{1}{4}(\alpha_\rho^1-\beta_{\rho,t}^1)+\frac{1}{4}(\alpha_\rho^2-\beta_{\rho,t}^2)+\nu\frac{1}{2}(\alpha_\rho^3 -\beta_{\rho,t}^3)
        \leq (\frac{1}{4}+\frac{1}{4}+\nu \frac{1}{2}) b \cdot (t-d_\rho)
        \leq b \cdot (t-d_\rho).
        \]
        This proves that in this case Constraints~\eqref{const JRP_D:holding} and~\eqref{const JRP_D:backlog} hold for $\rho$ and $t$.

        The other case is when $\GC(\rho,S_i)$ is called in Dual Assignment $i$ for some $i$.
        By Lemma~\ref{lem:global charge arrival}, we know that $\rho$ has arrived after $t_{i-1}$, i.e., $t_{i-1}<a_\rho$, which means that it was not involved in any of the dual assignments before $i$.
        The routine $\GC(\rho,S_i)$ is only called when $v \in V_{i,1}$, which means that $\rho$ could be involved in at most one charging before $\GC(\rho,S_i)$, which is $\LC(v,S_i)$ in step~\ref{D1} in Dual Assignment $i$.
        In this case, before calling $\GC(\rho,S_i)$, $\alpha_\rho$ and $\beta_{\rho,t}$ are set to $\frac{1}{4}\alpha_\rho^*$ and $\frac{1}{4}\beta_{\rho,t}^*$ obtained by Lemma~\ref{lem:local charging}.
        Lemma~\ref{cond:4} ensures that $\alpha_\rho^*$ and $\beta_\rho^*$ satisfy Constraints~\eqref{const JRP_D:holding} and~\eqref{const JRP_D:backlog}.
        
        Since $\GC(\rho,S_i)$ was called, it means that $\rho \in R_i^{v}\subseteq B_{i,1}^v$, which implies that $d_\rho \leq t_i$.
        So we have $t_{i-1} < a_\rho \leq d_\rho \leq t_i$
        
        In $\GC(\rho,S_i)$, $\alpha_\rho$ and $\beta_{\rho,t}$ for each $a_\rho \leq t \leq t_i$ are increased by the same amount (even if these increases are later scaled down in step~\ref{D223}), which ensures that Constraint~\eqref{const JRP_D:holding} holds if $a_\rho \leq t \leq d_\rho$.
        By the same reasoning, Constraint~\eqref{const JRP_D:backlog} holds for each $\rho$ and $d_\rho \leq t \leq t_i$.
        By design of the $\GC$ function, the final value for $\alpha_\rho$ after the global charging is done is at most $b \cdot (t_i-d_\rho)$. 
        Also, as we will demonstrate momentarily, $\beta_{\rho,t}\geq 0$ for each $t$.
        So for all $t>t_i$ we have:
        \[ \alpha_\rho - \beta_{\rho,t} \leq b \cdot (t_i-d_\rho) - 0 \leq b \cdot (t-d_\rho),\]
        which shows that Constraint~\eqref{const JRP_D:backlog} also holds for $t > t_i$.       
        \item Constraints \eqref{const JRP_D:beta nonnegativity} and \eqref{const JRP_D:gamma nonnegativity}: Since we initially set all the variables to 0 and in the process of constructing the dual solution we only increase variables, these constraints are implied.
    \end{itemize}
\end{proof}


\section{Other Related Work On Aggregation and Inventory Routing Problems}
\label{sec:rel_app}

JRP is closely related to different Online Aggregation and Inventory Routing problems, which have been studied extensively. 

\paragraph{Problems with only delay or backlog costs.} First we review problems with only backlog or delay costs that generally go under the terminology of (online) aggregation problems~\cite{chrobak2014online}. While in JRP, the underlying graph is a two-level tree, the other versions have considered different topologies.

\emph{Multi-Level Tree.}
 When the underlying graph is a tree and we only have backlog (waiting) cost, the problem is known as Multi-Level Aggregation Problem (MLAP), first introduced by Bienkowski et al.~\cite{bienkowski2020online}.
Aggregation problems for trees of arbitrary depth arise in multicasting, sensor networks, communication in organization hierarchies, and supply chain management. 
In~\cite{bienkowski2020online}, an $O(D^42^D)$-competitive algorithm is provided in the online case, where $D$ is the depth of the tree. This ratio is later improved to $O(D^2)$ by Azar and Touitou~\cite{azar2019general}.
Let MLAP-L and MLAP-D denote the variants of MLAP with, respectively, linear waiting cost functions and with hard deadlines. In~\cite{bienkowski2020online}, a $D^22^D$-competitive algorithm is presented for MLAP-D, which was improved later to $6(D+1)$ by Buchbinder et al.~\cite{buchbinder2017depth}. Currently, the best competitive ratio for MLAP-D is $D$~\cite{mcmahan2021d}.
Mari et al.~\cite{mari2024online} study MLAP-L in the stochastic setting, where the requests follow a Poisson arrival process. They provide a deterministic online algorithm that achieves a constant ratio of expectations.
MLAP has also been studied for specific tree topologies, e.g., paths.
Bienkowski et al.~\cite{bienkowski2013online} showed that the competitive ratio of this problem is between 3.618 and 5, and later, Bienkowski et al.~\cite{bienkowski2021new} improved it to 4 for MLAP-D. They also proved a lower bound of 4 even for the MLAP-L case.

\emph{General Metric.} Suppose the underlying graph represents a general metric, and the server can move in the graph to serve the requests. The problem of minimizing the sum of the distance traveled by the server and delay cost is called Online Service with Delay (OSD), first introduced by Azar et al.~\cite{azar2021online}. They give an $O(\log^4 n)$-competitive algorithm for the OSD problem, where $n$ is the number of nodes in the graph. This result was subsequently improved to $O(\log^2 n)$ by Azar and Touitou~\cite{azar2019general}. Both algorithms in~\cite{azar2021online, azar2019general} are randomized. Touitou~\cite{touitou2023improved} presented the first deterministic algorithm for the OSD problem, achieving $O(\log n)$ competitive ratio, which is the best currently known ratio for this problem.

\paragraph{Problems with only holding costs and hard deadlines.}

The classical formulation of JRP falls in this category and includes work surveyed earlier.

\emph{Multi-Level Trees.} Cheung et al.~\cite{cheung2016submodular} 
study the offline problem for multi-level trees with only holding costs. They call the problem tree JRP,
and provide a 3-approximation algorithm for it.

\emph{General Metric.} Suppose there is a general metric and only holding is allowed. This problem in the offline setting is known as the Inventory Routing Problem (IRP). The planning horizon can be finite or infinite.
IRP is a classical problem in Supply Chain Optimization and has a vast literature and many different variants~\cite{federgruen1995analytical, bramel1995location, sarmiento1999review, toth2002vehicle, campbell1998inventory}. No optimal or constant-factor approximations are known in the general case of the problem. It is not difficult to design an $O(\log n)$-approximation algorithm for this problem by reducing to the instances with tree metrics using the metric embedding
technique~\cite{fakcharoenphol2003tight}. When there are $T$ periods in the planning horizon, Nagarajan and Shi~\cite{nagarajan2016approximation} give a $\frac{\log T}{\log \log T}$-approximation algorithm. 
The approximation ratio was later improved to $O(\log \log \min(n,T))$ by Bosman and Olver~\cite{bosman2020improved}.
If we restrict ourselves to periodic schedules, Fukunaga et al.~\cite{fukunaga2014deliver} provide a 2.55-approximation.

\section{Pathological Example for Non-uniform Cost Functions}
\label{app:nonuniform}

In this section, we demonstrate why Algorithm~\ref{alg:single-item} (and as a result, its extension Algorithm~\ref{alg:multiple-item}) does not give a bounded competitive ratio for the case where the cost functions are non-uniform, i.e., each request $\rho$ has its own holding and backlog cost functions $h_\rho(.)$ and $b_\rho(.)$. 

The idea behind the pathological example is the following. After  each service is triggered by the algorithm, there might be a set of active requests with a small holding cost and 0 backlog cost each that exhaust the budget in the holding phase of Algorithm~\ref{alg:single-item}.  This will prevent the algorithm from satisfying the next active request that has a relatively large backlog cost and will trigger a service soon.

More formally, consider the following instance of the problem: The cost functions for each request $\rho$ are linear, with holding and backlog rates $h_\rho$ and $b_\rho$, respectively. This means for each request $\rho$, the holding cost from time $t<d_\rho$ to $d_\rho$ is $h_\rho \cdot (d_\rho-t)$, and the backlog cost from $d_\rho$ to $t>d_\rho$ is $b_\rho \cdot (t-d_\rho)$.
Let the service cost be $s=1$. Fix a large positive integer $N$. 
For each $i=0,1,\ldots,N$, there is a request $\rho_i$ with deadline $i+0.5$, holding rate $h_{\rho_i}=1/N^2$, and backlog rate $b_{\rho_i}=2$.
Also for each $i=1,\ldots,N$, there is a set of $N^3$ requests $R_i$ with deadlines $i+0.1$, holding rate $h_{R_i}=\frac{10}{N^3}$, and backlog rate $b_{R_i}=0$.
All of the requests have arrival time 0.

Algorithm~\ref{alg:single-item} accumulates $b_{\rho_0} \cdot (1-d_{\rho_0}) = 2 \cdot (1-0.5)=1$ backlog cost because of request $\rho_0$ at time 1, and therefore it triggers a service at time 1. 
Then in the holding phase, it satisfies the requests in $R_1$, paying $|R_1| \cdot h_{R_1}\cdot (1.1-1) = N^3 \cdot \frac{10}{N^3} \cdot 0.1 = 1$. 
Similarly, for each $i=1,\ldots,N$, the algorithm triggers a service at time $i+1$ because of request $\rho_i$, and it satisfies the requests in $R_i$ in the holding phase. 
This way, Algorithm~\ref{alg:single-item} triggers $N+1$ services, and pays $\Theta(N)$ cost.

However, one can trigger a service at time 0, satisfying all the requests $\rho_0,\ldots,\rho_N$ with a holding cost $\sum_{i=0}^N \frac{1}{N^2}\cdot (i+0.5)=O(1)$, and another service at time $N+1$ satisfying all the requests in $\bigcup_{i=1}^N R_i$ with backlog cost 0. 
The cost of this solution is $O(1)$.
Taking $N$ to $\infty$, we conclude that the competitive ratio of Algorithm~\ref{alg:single-item} is unbounded for non-uniform cost functions.

\end{sloppypar}
\end{document}